\documentclass[11pt]{amsart}

\usepackage[letterpaper, hmargin=1.25in, vmargin=1in]{geometry}
\usepackage{cite}
\usepackage{amsmath, amsfonts, amssymb}
\usepackage{mathtools, mathrsfs}
\usepackage{etoolbox}
\usepackage{algorithm}
\usepackage{algorithmic}
\usepackage{eqparbox}

\usepackage{graphicx}
\usepackage{textcomp}
\usepackage{xcolor}
\usepackage{tikz}
\usetikzlibrary{shapes,arrows,positioning}
\usepackage{comment}
\usepackage[hypertexnames = false]{hyperref}
\hypersetup{pdftitle={A Performance Improvement of the Payne--Hanek Range Reduction Algorithm}, pdfauthor={Tue Ly}}
\usepackage{float}
\usepackage{booktabs}

\newtoggle{draftmode}
\togglefalse{draftmode}

\let\originalleft\left
\let\originalright\right
\renewcommand{\left}{\mathopen{}\mathclose\bgroup\originalleft}
\renewcommand{\right}{\aftergroup\egroup\originalright}
\everymath{\displaystyle}
\allowdisplaybreaks

\newcommand{\R}{{\mathbb{R}}}

\newcommand{\Z}{{\mathbb{Z}}}

\newcommand{\F}{{\mathbb{F}}}

\newcommand{\PiInv}{\operatorname{PI\_INV}}

\newcommand{\abs}[1]{\left| {#1} \right|}

\newcommand{\prt}[1]{\left( {#1} \right)}
\newcommand{\brk}[1]{\left[ {#1} \right]}
\newcommand{\cbrk}[1]{\left\{ {#1} \right\}}
\newcommand{\floor}[1]{\left\lfloor {#1} \right\rfloor}
\newcommand{\ceil}[1]{\left\lceil{#1} \right\rceil}
\newcommand{\nearestint}[1]{\left\lfloor {#1} \right\rceil}

\makeatletter
\newcommand{\ufp}{\@ifnextchar\bgroup{\@ufpwitharg}{\operatorname{ufp}}}
\newcommand{\@ufpwitharg}[1]{\operatorname{ufp}\prt{#1}}

\newcommand{\uls}{\@ifnextchar\bgroup{\@ulswitharg}{\operatorname{uls}}}
\newcommand{\@ulswitharg}[1]{\operatorname{uls}\prt{#1}}

\newcommand{\ulp}{\@ifnextchar\bgroup{\@ulpwitharg}{\operatorname{ulp}}}
\newcommand{\@ulpwitharg}[1]{\operatorname{ulp}\prt{#1}}

\newcommand{\sinpi}{\@ifnextchar\bgroup{\@sinpiwitharg}{\operatorname{sinpi}}}
\newcommand{\@sinpiwitharg}[1]{\operatorname{sinpi}\prt{#1}}

\newcommand{\cospi}{\@ifnextchar\bgroup{\@cospiwitharg}{\operatorname{cospi}}}
\newcommand{\@cospiwitharg}[1]{\operatorname{cospi}\prt{#1}}

\newcommand{\fts}{\@ifnextchar\bgroup{\@ftswitharg}{\operatorname{Fast2Sum}}}
\newcommand{\@ftswitharg}[1]{\operatorname{Fast2Sum}\prt{#1}}

\newcommand{\fma}{\operatorname{fma}}
\newcommand{\prc}{\operatorname{prec}}
\makeatother

\numberwithin{equation}{section}
\newtheorem{Thm}{Theorem}
\newtheorem{Prop}[Thm]{Proposition}
\newtheorem{Cor}[Thm]{Corollary}
\newtheorem{Lem}[Thm]{Lemma}

\theoremstyle{definition}
\newtheorem{Rem}[Thm]{Remark}
\newtheorem{Def}[Thm]{Definition}

\newcommand{\xqed}[1]{\leavevmode\unskip\penalty9999 \hbox{}\nobreak\hfill\quad\hbox{#1}}
\newcommand{\RemEnd}{\xqed{$\triangle$}}

\newcommand{\eq}[2]{%
    \ifstrempty{#1}{%
        \begin{equation*} {#2} \end{equation*}%
    }{%
        \iftoggle{draftmode}{ {\tt ~[#1]~} }{\ignorespaces}%
        \begin{equation} \label{eq:#1} {#2} \end{equation}%
    }%
}
\newcommand{\eqr}[1]{\eqref{eq:#1}}
\newcommand{\lab}[1]{\label{#1} \iftoggle{draftmode}{{\tt [#1]}}{\ignorespaces}}

\makeatletter
\def\ps@headings{\ps@empty
  \def\@evenhead{%
    \setTrue{runhead}%
    \normalfont\scriptsize\hfil
    \def\thanks{\protect\thanks@warning}%
    \leftmark{}{}\hfil}%
  \def\@oddhead{%
    \setTrue{runhead}%
    \normalfont\scriptsize\hfil
    \def\thanks{\protect\thanks@warning}%
    \rightmark{}{}\hfil}%
  \def\@oddfoot{\normalfont\normalsize\hfil\thepage\hfil}%
  \let\@evenfoot\@oddfoot
  \let\@mkboth\markboth
}
\def\ps@firstpage{\ps@plain
  \def\@oddfoot{\normalfont\normalsize\hfil\thepage\hfil
     \global\topskip\normaltopskip}%
  \let\@evenfoot\@oddfoot}
\makeatother
\begin{document}

\title[Performance Improvement of the Payne--Hanek Range Reduction]{A Performance Improvement of the Payne--Hanek Range
    Reduction Algorithm}

\author{Tue Ly}
\address{Google LLC}
\email{lntue@google.com}

\subjclass[2020]{65G50, 65D20, 65Y04}
\keywords{Floating-point arithmetic, elementary functions, range reduction, Payne--Hanek algorithm.}

\begin{abstract}
Range reduction plays a crucial role in the accuracy and performance of evaluating trigonometric functions,
and is often the primary bottleneck for large floating-point inputs. While fast algorithms such as Cody--Waite
work efficiently over narrow intervals, the Payne--Hanek algorithm remains the standard technique for accurate
reduction across large floating-point inputs. However, many existing implementations of Payne--Hanek suffer from
high latency due to heavy branching, conversion overheads, and the use of multi-word integer arithmetic, which
hinders SIMD vectorization. In this paper, we analyze and present a branch-free variation of the Payne--Hanek
algorithm using only floating-point arithmetic. Our method operates directly over large double-precision inputs
($|x| \ge 2^{16}$) and is well suited to hardware with FMA instructions. We formulate the precision
constraints in terms of a truncation error budget, construct a compact lookup table indexed by the input
exponent, and prove that the scaled reduced argument has absolute error below $2^{-110}$ and relative error
below $2^{-60}$ for every input, including worst cases. The same routine can serve both as the complete range
reduction of a single-stage implementation and as the fast path of a correctly rounded one, achieving higher
throughput than existing implementations and lower latency than those returning a double-double reduced
argument. The algorithm is currently implemented in the LLVM libc project.
\end{abstract}

\maketitle

\section{Introduction} \lab{sec:Intro}

Evaluating trigonometric functions for a given floating-point input $x$ begins with a \emph{range reduction}
step, in which $x$ is decomposed as:
\eq{RedDec}{x = k \cdot \frac{\pi}{2^N} + y}
for an integer modulus scaling $N \ge 0$, an integer quotient $k \in \Z$, and a reduced argument $y \in \R$ satisfying:
\eq{YBound}{|y| \leq \frac{\pi}{2^{N + 1}} < 2^{-(N - 1)}.}

Using angle summation identities, evaluating trigonometric functions of $x$ reduces to evaluating those of $k
\cdot \frac{\pi}{2^N}$ and $y$:
\eq{Eval1}{
\begin{aligned}
    \sin(x) &= \sin \prt{k \cdot \frac{\pi}{2^N} + y} \\
            &= \sin \prt{k \cdot \frac{\pi}{2^N}} \cdot \cos y + \cos \prt{k \cdot \frac{\pi}{2^N}} \cdot \sin y.
\end{aligned}
}
In particular, the $k$ and $y$ can be defined by:
\eq{KY}{k = \nearestint{x \cdot \frac{2^N}{\pi}} \quad \text{and} \quad y = x - k \cdot \frac{\pi}{2^N},}
where $\nearestint{\cdot}$ denotes rounding to the nearest integer.

When $k$ is small, $y = x - k \cdot \frac{\pi}{2^N}$ can be computed to high precision using fast algorithms
such as Cody--Waite~\cite{C, CW} or Boldo et al.~\cite{BDL}. However, when $k$ is large, computing~\eqr{KY}
directly would require high intermediate precision because $\pi$ is irrational, causing catastrophic
cancellation if evaluated naively.

To resolve this issue, Payne and Hanek~\cite{PH} introduced an algorithm that exploits periodicity to perform
range reduction using moderate intermediate precision. Their main idea is to evaluate $y$ modulo $2\pi$:
\eq{KY1}{
\begin{aligned}
    y &= \prt{x - k \cdot \frac{\pi}{2^N}} \bmod 2\pi \\
      &= \frac{\pi}{2^N} \brk{ \prt{x \cdot \frac{2^N}{\pi} - k} \bmod 2^{N+1} },
\end{aligned}
}
where for any $K > 0$, $a \bmod K \in (-K/2, K/2]$ denotes the centered modulo operation.
This formulation allows all bits of $x \cdot \frac{2^N}{\pi}$ with weights $2^{N+1}$ and higher to be discarded.

Specifically, for an input $x$, we express $\frac{2^N}{\pi}$ as a sum of floating-point limbs representable in $\F$:
\eq{TwoNOverPi}{
    \frac{2^N}{\pi} = \cdots + C_{-1} + C_0 + C_1 + C_2 + \cdots
}
where the limbs can be signed when generated using round-to-nearest ($\operatorname{RN}$), with
$\ufp{C_i} < \uls{C_{i - 1}}$ (see Section~\ref{sec:NotCon} for the definitions of $\ufp$ and $\uls$ functions).
It is important to note that this decomposition is not fixed once and for all, but rather depends on the
magnitude of $x$ (or its exponent range, as detailed in Section~\ref{sec:LUT}) so that:
\eq{LimbCond}{
    \uls{C_{-1}} \geq \frac{2^{N + 1}}{\ulp{x}} > \uls{C_0}.
}
Consequently, all negatively indexed terms $x \cdot C_i$ ($i < 0$) are integer multiples of $2^{N + 1}$ and
can be dropped modulo $2^{N + 1}$:
\eq{KY2}{
    y = \frac{\pi}{2^N} \brk{\prt{x \cdot \prt{C_0 + C_1 + \cdots} - k} \bmod 2^{N + 1}}.
}

In practical computation, the infinite expansion in~\eqr{KY2} must be truncated to a small, finite number of
limbs $C_0, \dots, C_{L-1}$. Because floating-point numbers are discrete and $\pi$ is irrational, $x \cdot
\frac{2^N}{\pi}$ is bounded away from all integers for every input with $|k| \ge 1$. The minimum of this distance
over all such inputs is known as the \emph{worst-case distance of range reduction} $w^*$ (cf.~\cite{Mu1, Mu2}),
and it allows the trailing limbs to be dropped as long as their total truncation error is small enough.
Section~\ref{sec:WorstCase} derives the number of limbs $L$ needed for all inputs, and for the fast path of a
correctly rounded implementation.

Prominent implementations of the Payne--Hanek algorithm include \emph{fdlibm} \cite{F, Ng} and
\emph{libultim}~\cite{I} (which underlie FreeBSD, and glibc through version~2.43, respectively), as well as correctly
rounded libraries such as CRLIBM~\cite{CR} and the CORE-MATH project~\cite{CM, Zim26}. While accurate, classical
multi-word implementations are significantly slower than reduction for small $k$. CORE-MATH improved performance by
using the exponent field of $x$ to index the required chunk of bits of $\frac{1}{2\pi}$ (an approach also discussed by
Markstein~\cite{Ma} and Harrison~\cite{Har00} for IA-64); however, it relies on multi-word integer arithmetic and
cross-domain integer-to-floating-point conversions. Furthermore, earlier implementations (cf.~\cite{F, I, CR, RLIBM})
are heavily branched and not friendly to SIMD vectorization.

Shibata and Petrogalli~\cite{SLEEF} introduced a vectorized variation of the Payne--Hanek algorithm in the
SLEEF library, which also performs range reduction using floating-point arithmetic rather than multi-word
integers. However, SLEEF indexes its reduction table by individual exponent binades ($M = 1$), requiring a
full per-exponent 969-entry table ($\approx 30$~KB) that occupies most of the L1 data cache. In addition, their
formulation is specialized to $N = 1$ (modulo $\frac{\pi}{2}$), yielding $\abs{y} \le \frac{\pi}{4}$, which requires a
higher-degree approximation polynomial (or quadrant sign folding) but avoids a secondary trigonometric lookup
table. In contrast, our approach groups exponents into blocks of $M = 16$, reducing the table footprint by
a factor of 16 down to 2~KB without branching. Furthermore, we generalize to arbitrary moduli $N$
(such as $N = 7$ for $\frac{\pi}{128}$) and establish provable error bounds from floating-point limb
constraints.

In the implementations considered here, the output of range reduction is $k \bmod 2^{N+1}$ and a double-double
approximating either:
\eq{YD}{
    \hat{y} = y_h + y_\ell \approx \frac{\pi}{2^N} \brk{\prt{x \cdot \frac{2^N}{\pi} - k} \bmod 2^{N + 1}},
}
or
\eq{UD}{
    \hat{u} = u_h + u_\ell \approx \prt{x \cdot \frac{2^N}{\pi} - k} \bmod 2^{N + 1}.
}

\begin{Rem}
    If the output is of the form~\eqr{UD}, evaluation can proceed via $\sinpi{z} = \sin(\pi z)$ and $\cospi{z}
    = \cos(\pi z)$:
    \eq{Eval2}{
    \begin{aligned}
        \sin(x) &= \sin \brk{\frac{\pi}{2^N} \prt{k + u}} \\
                &= \sinpi{\frac{k}{2^N}} \cdot \cospi{\frac{u}{2^N}} + \cospi{\frac{k}{2^N}} \cdot
                    \sinpi{\frac{u}{2^N}}.
    \end{aligned}
    }
    This eliminates an explicit multiplication by $\frac{\pi}{2^N}$ at the end of range reduction. Note that
    when $k=0$ (small $x$), using $\sinpi$ and $\cospi$ is not advisable because $\sin(x) \approx x$ is
    evaluated directly by a simple polynomial, whereas $\sinpi(x/\pi)$ would introduce unnecessary rounding
    overhead.
    \RemEnd
\end{Rem}

\begin{Rem}
    Historically, most implementations chose $N = 1$ (reduction modulo $\frac{\pi}{2}$) because $\sin\prt{k \cdot
    \frac{\pi}{2} + y}$ simplifies to $\pm \sin(y)$ or $\pm \cos(y)$ with simple quadrant logic:
    \eq{}{
    \sin(x) = \begin{cases}
        \sin(y), & k \equiv 0 \pmod{4}, \\
        \cos(y), & k \equiv 1 \pmod{4}, \\
        -\sin(y), & k \equiv 2 \pmod{4}, \\
        -\cos(y), & k \equiv 3 \pmod{4}.
    \end{cases}
    }
    The choice of the reduction modulus $N$ (Table~\ref{tab:RangeReductionN}) is a trade-off between lookup table size
    and the computational cost of the reduced evaluation, particularly in correctly rounded libraries:
    \begin{itemize}
        \item \textbf{Target precision vs.\ intermediate precision in Ziv's test:}
        Because intermediate computations are executed in double precision (the same as the target precision),
        a fast path must produce its result in a double-double format $y_h + y_\ell$ with an overall relative
        error well below $2^{-53}$ for Ziv's rounding test to succeed with high probability.
        
        \item \textbf{Small $N$:}
        Choosing $N = 1$ without a secondary table (as in FreeBSD~\cite{F} and SLEEF~\cite{SLEEF}) eliminates
        the trigonometric reconstruction table entirely. However, the reduced argument remains relatively large
        ($\abs{y} \le \frac{\pi}{4}$); in a correctly rounded fast path, this requires evaluating a high-degree
        polynomial (typically degree 15--21) in double-double arithmetic across many terms to obtain the low part
        $y_\ell$ with sufficient accuracy.
        
        \item \textbf{Larger $N$:}
        Choosing a larger modulus, such as $N = 7$ (in LLVM libc~\cite{L}), $N = 8$ (in CRLIBM~\cite{CR}), or
        $N = 10$ and $N = 14$ (in CORE-MATH~\cite{CM, Zim26}), restricts the reduced argument to $\abs{y} \le
        \frac{\pi}{2^{N+1}}$ (e.g., $\abs{y} \le \frac{\pi}{256}$ for $N = 7$). Because the Taylor series converges
        rapidly on this narrow interval, the higher-order terms (such as $\frac{y^3}{6}, \frac{y^5}{120}, \dots$) are
        much smaller than the leading term $y_h$. Consequently, these terms can be evaluated in double precision using
        only a few FMAs and accumulated directly to form the low part $y_\ell$, without requiring multi-word polynomial
        arithmetic. This enables a fast Ziv test ($> 99.99\%$ pass rate) at the expense of a small secondary lookup
        table for $\sin\prt{k \cdot \frac{\pi}{2^N}}$ and $\cos\prt{k \cdot \frac{\pi}{2^N}}$. For LLVM libc ($N=7$),
        it is 4~KB for 256 double-double entries in the default build, or about 1~KB with 65 entries under a
        small-table build; CRLIBM ($N=8$) uses a 2~KB table of 65 double-double $\sin/\cos$ pairs via octant symmetry;
        CORE-MATH's \texttt{reduce\_fast} ($N=10$) uses a 6~KB table of 256 triples (a small offset $x_i - i/2^{11}$
        and double values of $\sin 2\pi x_i$ and $\cos 2\pi x_i$, accurate to about 68 bits) via octant symmetry; and
        CORE-MATH's newer \texttt{reduce\_large} ($N=14$~\cite{Zim26}) uses an 8~KB bipartite table ($i_1\pi/2^7$
        and $i_2\pi/2^{14}$) so that the reduced argument fits in a single \texttt{double}.
    \end{itemize}
    By grouping exponents ($M = 16$), our algorithm reduces the primary Payne--Hanek table from 32~KB down to
    only 2~KB. Combined with the secondary table, the total memory footprint in LLVM libc
    ($\approx 3-6$~KB) is $5$--$10$ times smaller than per-exponent tables like SLEEF's (969 entries,
    $\approx 30$~KB). Furthermore, because our framework supports arbitrary $N$, it allows implementations to
    choose the modulus $N$ that best balances memory footprint and arithmetic performance on their target
    architecture. Unlike CORE-MATH's multi-word integer reductions (\texttt{reduce\_fast} and \texttt{reduce\_large},
    which rely on Ziv's test and an accurate path for cancellation inputs), Algorithm~\ref{alg:main} executes in
    floating-point FMA arithmetic without branches and meets the worst-case accuracy budget ($< 2^{-110}$ vs.\
    $< 2^{-73.3}$ and $< 2^{-66.3}$).
    \begin{table}[htbp]
        \caption{Reduction Moduli and Secondary Table Sizes} \lab{tab:RangeReductionN}
        \begin{center}
        \footnotesize
        \setlength{\tabcolsep}{2.2pt}
        \begin{tabular}{ccccccccc}
        \toprule
            & \textbf{Brisebarre} & \textbf{FreeBSD} & \textbf{glibc} & \textbf{CRLIBM} &
            {\begin{tabular}[c]{@{}c@{}}\textbf{CORE-MATH}\\(\texttt{fast})\end{tabular}} &
            {\begin{tabular}[c]{@{}c@{}}\textbf{CORE-MATH}\\(\texttt{large})\end{tabular}} &
            \textbf{SLEEF} &
            {\begin{tabular}[c]{@{}c@{}}\textbf{LLVM}\\\textbf{libc}\end{tabular}} \\
        \midrule
        $N$ & 1 & 1 & 1 & 8 & 10 & 14 & 1 & 7 \\
        {\begin{tabular}[c]{@{}c@{}}Fast\\Table\end{tabular}} &
            N/A & None & $\sim 3.5$ KB & 2 KB & 6 KB & 8 KB & None & 1--4 KB \\
        {\begin{tabular}[c]{@{}c@{}}Acc.\\Table\end{tabular}} &
            N/A & N/A & N/A & None & 16 KB & 2.5 KB & N/A & 1.5 KB \\
        \bottomrule
        \end{tabular}
        \end{center}
        \vspace{1ex}
        \raggedright{\footnotesize Here $N$ denotes the Payne--Hanek reduction modulus $\pi/2^N$, and the two rows
        report the secondary $\sin/\cos$ lookup table sizes in the fast and accurate (Ziv fallback) paths. FreeBSD,
        glibc, and SLEEF are single-stage (glibc's second stage rounds the reduced argument to a multiple of
        $1/128$~rad); CRLIBM's accurate path reduces modulo $\pi/2$ without a secondary table.}
    \end{table}    
    \RemEnd
\end{Rem}

In this paper, we present a branch-free floating-point variation of the Payne--Hanek range reduction algorithm
for large double-precision inputs. By grouping exponents to use a compact 2~KB lookup table and supporting
arbitrary moduli $N$ (such as $N=7$ for $\pi/128$), our algorithm computes double-double outputs as
in~\eqr{YD} and~\eqr{UD} entirely without branches, achieving competitive latency and throughput while guaranteeing
provable error bounds. Its truncation error meets the worst-case budget of Section~\ref{sec:WorstCase}, so the
reduced argument stays accurate even for the inputs closest to multiples of $\frac{\pi}{2^N}$. The same routine
can therefore be used as the complete range reduction of a single-stage implementation, without cancellation
checks or an accurate path, and as the fast path of a two-stage correctly rounded implementation, whose accurate
path can reuse its exact partial products.

\subsection{Algorithm at a Glance and Intuition} \lab{subsec:Glance}

Before presenting the detailed mathematical foundations and error bounds, we give an overview of the algorithm
and how it is structured.

Standard Payne--Hanek implementations (such as those in fdlibm~\cite{F}, glibc~\cite{I},
CRLIBM~\cite{CR}, and CORE-MATH~\cite{CM}) often rely on multi-word arithmetic or multi-word integer conversions to
compute intermediate digits of $x \cdot 2^N/\pi$. While effective, using multi-word integers has a few practical
drawbacks on modern hardware:
\begin{itemize}
    \item \textbf{Branches and serialization:} Propagating integer carries and extracting bits often requires
        conditional branches and sequential steps.
    \item \textbf{Register conversions:} Converting the resulting multi-word integer back into a
        floating-point double-double ($u_h + u_\ell$) requires moving data between general-purpose integer
        registers and floating-point or vector registers.
    \item \textbf{SIMD vectorization:} Large per-exponent lookup tables or multi-word integer carry chains
        make vectorization difficult.
\end{itemize}

In contrast, our algorithm performs the reduction entirely using floating-point arithmetic and hardware fused
multiply-add (FMA) instructions. By grouping exponents to index a small lookup table and choosing the limb
alignments carefully, the integer quotient and the reduced argument can be computed without branches (as
discussed in Remark~\ref{rem:integer_table}, this principle also extends to integer fixed-point
implementations).

The algorithm consists of three main phases, as shown in the dataflow diagram in Figure~\ref{fig:pipeline}:
\begin{itemize}
    \item \textbf{Phase 1: Exponent grouping and input scaling.}
    Instead of storing an entry for each individual exponent (which would require a 32~KB table), we group
    exponents into blocks of $M=16$. For an input $x = (-1)^{s_x} 2^{e_x} m_x$, the table index $i =
    \floor{(e_x - 62)/16}$ is computed by shifting the exponent. We load four double-precision limbs $D_0,
    D_1, D_2, D_3$ from a 2~KB table $\PiInv$ (\texttt{ONE\_TWENTY\_EIGHT\_OVER\_PI} in LLVM libc), and scale
    $x$ down by $2^{-16i}$ to $|x_r| \in [2^{62}, 2^{78})$ using an integer subtraction on the exponent.
    
    \item \textbf{Phase 2: Parallel products and quotient extraction.}
    Three exact products $x_r \cdot D_0, x_r \cdot D_1, x_r \cdot D_2$ are computed, yielding exact high and
    low pairs $(a_h, a_\ell)$, $(b_h, b_\ell)$, and $(c_h, c_\ell)$ respectively. By design, the leading
    product $a_h = (x_r \cdot D_0)_h$ satisfies $\ulp{a_h} \ge 2^{N+1}$; therefore, $a_h \equiv 0
    \pmod{2^{N+1}}$ and can simply be dropped.
    The binary point falls across $a_\ell$ and $b_h$. We compute the integer quotient $\hat{k} =
    \nearestint{a_\ell \oplus b_h}$ and the high part of the remainder $v = (a_\ell \ominus \hat{k}) \oplus
    b_h$. Because $a_\ell - \hat{k}$ fits in $p$ bits on the grid of multiples of $\eta = \min(\uls{a_\ell},
    1)$, the subtraction $a_\ell \ominus \hat{k}$ is exact, introducing no rounding error.
    
    \item \textbf{Phase 3: Tail accumulation.}
    The remaining terms $b_\ell, c_h, c_\ell, x_r D_3$ represent bits of lower significance. Since
    $\uls{b_\ell} \ge \ulp{c_h}$, Dekker's Fast2Sum~\cite{Dekker} is exact under round-to-nearest without a
    comparison~\cite{JZ25}: $q_h + q_\ell = \fts{b_\ell, c_h}$. The trailing term $x_r D_3$ is accumulated
    into $c_\ell$ with an FMA to form $r = \circ(c_\ell + x_r D_3)$. Combining $v$ and $q_h$ with another
    Fast2Sum and summing the low terms yields the reduced double-double residual $\hat{u} = u_h + u_\ell$ with
    error bounded by $2^{-105}$.
\end{itemize}

\begin{figure}[htbp]
\centering
\begin{tikzpicture}[
    scale=0.85, every node/.style={transform shape},
    box/.style={rectangle, draw, rounded corners=2pt, minimum height=1.7em, minimum width=4.5em,
        text centered, font=\footnotesize},
    tablebox/.style={rectangle, draw, fill=blue!8, rounded corners=2pt, minimum height=1.7em,
        text centered, font=\footnotesize},
    op/.style={circle, draw, fill=orange!15, inner sep=1.5pt, font=\scriptsize},
    res/.style={rectangle, draw, fill=green!10, rounded corners=2pt, minimum height=1.7em,
        text centered, font=\footnotesize\bfseries},
    arr/.style={-stealth, thick, draw=gray!80},
    bus/.style={thick, draw=gray!80}
]
    \node[box, fill=gray!10] (x) at (-1.8, 6.7) {$x$};
    \node[tablebox] (idx) at (1.8, 6.7) {Index $i = \floor{\frac{e_x - 62}{16}}$};

    \node[box, fill=gray!10] (xr) at (-1.8, 5.5) {$x_r = x \cdot 2^{-16i}$};
    \node[tablebox] (lut) at (1.8, 5.5) {$\operatorname{PI\_INV}[i] \to (D_0, D_1, D_2, D_3)$};

    \draw[arr] (x) -- (idx);
    \draw[arr] (x) -- (xr);
    \draw[arr] (idx) -- (lut);
    \draw[arr] (idx.south) to[out=210, in=30] (xr.north east);

    \draw[bus] (-4.5, 4.6) -- (4.5, 4.6);
    \draw[bus] (xr.south) -- (-1.8, 4.6);
    \draw[bus] (lut.south) -- (1.8, 4.6);

    \node[box, fill=purple!8] (p0) at (-4.5, 3.7) {$a_h + a_\ell = x_r D_0$};
    \node[box, fill=purple!8] (p1) at (-1.5, 3.7) {$b_h + b_\ell = x_r D_1$};
    \node[box, fill=purple!8] (p2) at (1.5, 3.7) {$c_h + c_\ell = x_r D_2$};
    \node[box, fill=purple!8] (p3) at (4.5, 3.7) {$x_r D_3$};

    \draw[arr] (-4.5, 4.6) -- (p0.north);
    \draw[arr] (-1.5, 4.6) -- (p1.north);
    \draw[arr] (1.5, 4.6) -- (p2.north);
    \draw[arr] (4.5, 4.6) -- (p3.north);

    \node[op] (kadd) at (-3.0, 2.6) {$\oplus$};
    \node[res] (k) at (-4.5, 1.5) {$\hat{k} = \nearestint{a_\ell \oplus b_h}$};
    \node[op] (vsub) at (-2.7, 1.5) {$\ominus$};
    \node[op] (vadd) at (-1.5, 1.5) {$\oplus$};
    \node[box, fill=yellow!15] (v) at (-1.5, 0.4) {$v = (a_\ell \ominus \hat{k}) \oplus b_h$};

    \draw[arr] (p0.south) to[out=270, in=120] (kadd);
    \draw[arr] (p1.south) to[out=270, in=60] (kadd);
    \draw[arr] (kadd) -- (k);
    \draw[arr] (k) -- (vsub);
    \draw[arr] (p0.south) to[out=270, in=150] (vsub);
    \draw[arr] (vsub) -- (vadd);
    \draw[arr] (p1.south) to[out=270, in=90] (vadd);
    \draw[arr] (vadd) -- (v);

    \node[box, fill=cyan!10] (f2s1) at (1.5, 2.5) {$\fts{b_\ell, c_h}$};
    \node[box, fill=yellow!15] (q) at (1.5, 1.4) {$q_h + q_\ell$};
    \draw[arr] (p1.south) to[out=270, in=140] (f2s1.north);
    \draw[arr] (p2.south) to[out=270, in=90] (f2s1.north);
    \draw[arr] (f2s1) -- (q);

    \node[op] (rtail) at (4.5, 2.5) {$\fma$};
    \node[box, fill=yellow!15] (r) at (4.5, 1.4) {$r = \circ(c_\ell + x_r D_3)$};
    \draw[arr] (p2.south) to[out=270, in=140] (rtail);
    \draw[arr] (p3.south) -- (rtail);
    \draw[arr] (rtail) -- (r);

    \node[box, fill=cyan!10] (f2s2) at (0.0, -0.5) {$\fts{v, q_h} \to (u_h, t_\ell)$};
    \draw[arr] (v.south) to[out=270, in=150] (f2s2.north);
    \draw[arr] (q.south) to[out=270, in=30] (f2s2.north);

    \node[res] (res) at (0.0, -1.7) {$\hat{u} = u_h + (t_\ell \oplus (q_\ell \oplus r))$};
    \draw[arr] (f2s2) -- (res);
    \draw[arr] (r.south) to[out=270, in=0] (res.east);

\end{tikzpicture}
\caption{Dataflow pipeline of the branch-free floating-point range reduction.} \lab{fig:pipeline}
\end{figure}
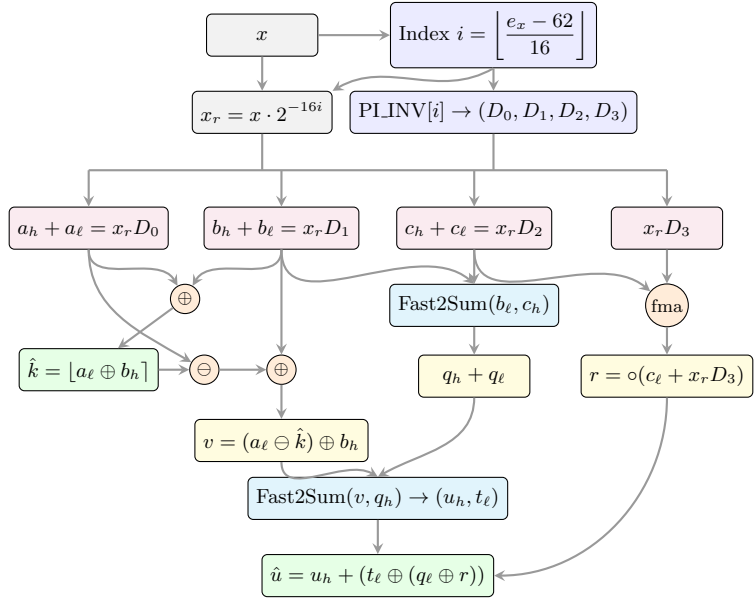

The remainder of this paper is organized as follows. Section~\ref{sec:NotCon} establishes notations and conventions.
Section~\ref{sec:WorstCase} derives the required number of limbs from a truncation error budget, for both
worst-case accuracy and the fast path of correctly rounded implementations. Section~\ref{sec:PH}
presents branch-free floating-point steps to compute $k$ and the high part of the residual. Section~\ref{sec:LUT}
details the exponent-indexed lookup table and memory hierarchy trade-offs. Section~\ref{sec:Main} details the complete
algorithm and proves its correctness. Section~\ref{sec:rounding} analyzes behavior under directed rounding modes
and architectures without FMA. Section~\ref{sec:perf} evaluates benchmark performance and SIMD considerations,
and Section~\ref{sec:Rem}
provides concluding remarks.

\section{Notations and Conventions} \lab{sec:NotCon}

The notations and conventions used throughout this paper are:
\begin{itemize}
    \item $\mathbb{D} = \mathbb{Z}\brk{\frac{1}{2}} = \cbrk{m \cdot 2^e \mid m, e \in \mathbb{Z}}$:
        the ring of dyadic numbers.
    \item $\F$: the finite set of normal floating-point numbers in a target format (e.g., IEEE 754 binary64).
    \item $p$: the precision (number of significand bits) of normal numbers in $\F$ ($p = 53$ for double precision).
    \item $\circ : \R \longrightarrow \F$:
        a rounding operator (defaulting to round-to-nearest, ties-to-even, $\operatorname{RN}$).
    \item $\oplus, \ominus, \otimes$:
        floating-point addition, subtraction, and multiplication in $\F$, where $a \oplus b := \circ(a + b)$.
\end{itemize}

Whereas terms such as $\text{msb}$ and $\text{lsb}$ designate bit positions in standards and general literature
(e.g., IEEE 754-2019~\cite{IEEE754}, Goldberg~\cite{G}, Higham~\cite{H}), precision analysis in the floating-point
literature (e.g., Rump et al.~\cite{ROO}, Muller et al.~\cite{Mu2}, and Joldes et al.~\cite{JMPT}) relies on the
unit in the first place ($\ufp$) and unit in the last significant place ($\uls$) to represent the weights of bits:
\begin{Def} \lab{def:UfpUls}
Let $x \in \mathbb{D}^* := \mathbb{D} \setminus \{0\}$ be a non-zero dyadic number. We define:
\begin{itemize}
    \item $\ufp{x} := 2^{\floor{\log_2{\abs{x}}}}$: the weight of the most significant bit of $x$.
    \item $\uls{x} := \max \cbrk{2^m \mid x \cdot 2^{-m} \in \mathbb{Z}}$:
        the weight of the least significant non-zero bit of $x$.
\end{itemize}
For completeness, we adopt the conventions $\ufp{0} := 0$, $\ulp{0} := 0$, and $\uls{0} := +\infty$.
\end{Def}

For a normal number $x \in \F$ of precision $p$, the unit in the last place is uniquely defined as:
\eq{}{\ulp{x} = \ulp_\F(x) = \ufp{x} \cdot 2^{-p + 1}.}

\begin{Lem} \lab{Lem:UfpUlsProp}
    Basic properties of $\ufp$ and $\uls$:
    \begin{itemize}
        \item[a.] A dyadic number $x \in \mathbb{D}^*$ with $|x|$ within the normal floating-point range is
            representable in $\F$ of precision $p$ if and only if $\uls{x} \geq \ufp{x} \cdot 2^{-p + 1}$.
        \item[b.] (Additivity) For any $x, y \in \mathbb{D}^*$ with $x + y \ne 0$:
            \eq{}{
                \begin{aligned}
                    \uls{x + y} &\geq \min \cbrk{\uls{x}, \uls{y}}, \\
                    \ufp{x + y} &\leq 2 \max \cbrk{\ufp{x}, \ufp{y}}, \\
                    \ufp{x + y} &\geq \max \cbrk{\ufp{x}, \ufp{y}} \quad \text{if } x \cdot y > 0.
                \end{aligned}
            }
        \item[c.] (Multiplicativity) For any $x, y \in \mathbb{D}^*$:
            \eq{}{
                \begin{aligned}
                    \uls{x \cdot y} &= \uls{x} \cdot \uls{y}, \\
                    \ufp{x} \cdot \ufp{y} &\leq \ufp{x \cdot y} \leq 2 \ufp{x} \cdot \ufp{y}.
                \end{aligned}
            }
    \end{itemize}
\end{Lem}

\section{Error budgets, worst-case distance of range reduction, and number of limbs} \lab{sec:WorstCase}

For an input format $\F$ of precision $p$, we define the worst-case distance of trigonometric range reduction by:
\eq{WorstCase}{
    w^* = w_N(\F) = \min_{\substack{z \in \F, |z| \ge \pi/2^{N+1} \\ j \in \mathbb{Z}}}
        \abs{z \cdot \frac{2^N}{\pi} - j}.
}
Restricting to $|z| \ge \pi/2^{N+1}$ excludes the trivial minimum: without this constraint, choosing $z \to 0$ and
$j = 0$ would yield 0.

\begin{Rem} \lab{Rem:WorstCaseScale}
In the literature (e.g., \cite{BDKMR, Mu1, Mu2}), the worst-case distance is often defined as:
\eq{}{v_N(\F) = \min_{\substack{z \in \F, |z| \ge \pi/2^{N+1} \\ j \in \Z}} \abs{z - j \cdot \frac{\pi}{2^N}}.}
A drawback of $v_N(\F)$ is its dependency on $N$, since $v_N(\F) = 2^{-N} v_0(2^N \F)$.
In contrast, $w_N(\F)$ satisfies:
\eq{}{
    w_N(\F) = \min_{z \in 2^N \F, j \in \mathbb{Z}} \abs{\frac{z}{\pi} - j} = w_0(2^N \F).
}
Except for the topmost $N$ binades where $2^N z$ exceeds $\text{DBL\_MAX}$, scaling by $2^N$ shifts the
exponent binades. An exhaustive search across the topmost binades confirms that for $N \le 15$, no new worst
cases occur, so $w_N(\F) = w_0(\F) = w^*$. This provides a single uniform constant $w^*$ to bound errors across
arbitrary reduction moduli.
\RemEnd
\end{Rem}

Recall from decomposition~\eqr{TwoNOverPi} that $\frac{2^N}{\pi} = \sum_{i=-\infty}^\infty C_i$.
To drop all limbs $C_i$ with $i < 0$ in the computation modulo $2^{N+1}$, it suffices that:
\eq{}{\uls{x \cdot C_{-1}} = \uls{x} \cdot \uls{C_{-1}} \geq 2^{N + 1}.}
To guarantee this condition, it suffices to have:
\eq{}{\uls{x} \cdot \ufp{C_0} \geq 2^N.}
Using the sharp lower bound $\uls{x} \geq \ufp{x} \cdot 2^{-p + 1}$, we obtain the sufficient condition:
\eq{C0Bound}{\ufp{x} \cdot \ufp{C_0} \geq 2^{N + p - 1}.}

Conversely, when using $L$ limbs $C_0, \dots, C_{L-1}$ with $\ufp{C_i} < \uls{C_{i-1}}$, the truncation error
modulo $2^{N+1}$ is bounded by:
\eq{TruncErr}{
\begin{aligned}
    \abs{x \cdot \prt{\frac{2^N}{\pi} - \sum_{i=-\infty}^{L-1} C_i} \bmod 2^{N+1}}
        &= \abs{x \cdot \sum_{i=L}^\infty C_i} \\
        &\leq 2 \abs{x} \cdot \ufp{C_L} < 4 \ufp{x} \cdot \ufp{C_L}.
\end{aligned}
}
This is an absolute error on the normalized residual $u = \prt{x \cdot \frac{2^N}{\pi} - k} \bmod 2^{N+1}$
approximated in~\eqr{UD}, measured in the same unit as $w^*$. To keep it below an absolute error budget
$\varepsilon > 0$, it suffices that:
\eq{SuffBounds}{
\begin{cases}
    \ufp{x} \cdot \ufp{C_0} \geq 2^{N + p - 1}, \\
    \ufp{x} \cdot \ufp{C_L} \leq 2^{-2} \varepsilon.
\end{cases}
}

If each limb $C_i$ has precision $p_c$ bits (so $\ufp{C_L} \leq \ufp{C_0} \cdot 2^{-L p_c}$), we
combine~\eqr{SuffBounds} into:
\eq{LimbBounds1}{
    2^{N + p - 1} \leq \ufp{x} \cdot \ufp{C_0} \leq \varepsilon \cdot 2^{L p_c - 2}.
}
This yields the resulting lower bound on the number of limbs $L$:
\eq{LimbBounds2}{L \geq \frac{N + p + 1 - \log_2 \varepsilon}{p_c}.}

Since $y = \frac{\pi}{2^N} \cdot u$ by~\eqr{KY1}, the corresponding error on $y$ is below
$\varepsilon \cdot \frac{\pi}{2^N}$. It stays below $\ulp{u}$ whenever $\ulp{u} \geq \varepsilon$, i.e., for a
power of two $\varepsilon$, whenever $|u| \geq 2^{p-1} \varepsilon$. The appropriate budget depends on how the
reduced argument is used:
\begin{itemize}
    \item \emph{Worst-case accuracy}, $\varepsilon = \ulp{w^*} = \ufp{w^*} \cdot 2^{-p+1}$: since
        $|u| \geq w^* \geq \ufp{w^*} = 2^{p-1} \varepsilon$ whenever $|k| \geq 1$, the truncation error stays
        below $\ulp{u}$ for every input that requires reduction. This is needed when the result is returned
        without a rounding test, e.g., in a single-stage implementation. Then~\eqr{LimbBounds2} becomes:
        \eq{LimbBoundsWC}{L \geq \frac{N + 2p - \log_2 \prt{\ufp{w^*}}}{p_c}.}
    \item \emph{Fast path of a correctly rounded implementation}, $\varepsilon = 2^{-p-14}$: in a two-stage
        implementation following Ziv's strategy, the rounding test of the fast path includes the absolute
        error $\varepsilon \cdot \frac{\pi}{2^N}$ on $y$ in its error bound. This error is negligible unless
        $x$ is close to a zero of $\sin$ or $\cos$, and the inputs for which it matters are caught by the test
        and resolved by the accurate path, so $\varepsilon$ can be far larger than $\ulp{w^*}$. Assuming,
        heuristically, that $u$ is uniformly distributed on $[-1/2, 1/2]$, the condition
        $|u| \geq 2^{p-1} \varepsilon = 2^{-15}$ fails for only about $2^{-14} \approx 6.1 \times 10^{-5}$ of
        the inputs. Then~\eqr{LimbBounds2} becomes:
        \eq{LimbBoundsFP}{L \geq \frac{N + 2p + 15}{p_c}.}
\end{itemize}
Compared with~\eqr{LimbBoundsWC}, the fast path replaces the $-\log_2 \prt{\ufp{w^*}}$ bits ($62$ for binary64
and $30$ for binary32) by $15$.

Table~\ref{tab:limbs} reports the minimal number of limbs required for $N \le 15$ across common precisions,
for both budgets. For binary64, the fast-path budget $\varepsilon = 2^{-67}$ is met by three limbs of the table
in Section~\ref{sec:LUT}, where~\eqr{MBoundEps} certifies budgets down to $2^{-69}$. The worst-case values
$w^*$ were computed via continued fractions of $2^N/\pi$ following~\cite{Mu2}. For binary64, the fast path saves
one limb regardless of the precise budget: $L_{\min} = 3$ for every $\varepsilon \geq 2^{-84}$. For binary32 with
$p_c = 24$, condition~\eqr{LimbBoundsFP} becomes:
\eq{}{
    L \ge \frac{N + 63}{24},
}
so the fast path saves one limb only for $N \le 9$, while the worst-case budget requires four limbs for all
$N \le 15$.

\begin{table}[htbp]
    \caption{Minimal Number of Limbs} \lab{tab:limbs}
    \begin{center}
    \begin{tabular}{ccccccc}
    \toprule
        & & & & & \multicolumn{2}{c}{$L_{\min}$} \\
    \cmidrule(lr){6-7}
        $p$ & $p_c$ & $N$ & $w^*$ & $\log_2 \prt{\ufp{w^*}}$ & worst case & fast path \\
            & & & & & ($\varepsilon = \ulp{w^*}$) & ($\varepsilon = 2^{-p-14}$) \\
    \midrule
        53 & 53 & $\le 15$ & $2.98394250375 \times 10^{-19}$ & $-62$ & 4 & 3 \\
        53 & 51 & $\le 15$ & $2.98394250375 \times 10^{-19}$ & $-62$ & 4 & 3 \\
        24 & 24 & $\le 9$ & $1.02799438139 \times 10^{-9}$ & $-30$ & 4 & 3 \\
        24 & 24 & $10$--$15$ & $1.02799438139 \times 10^{-9}$ & $-30$ & 4 & 4 \\
        24 & 53 & $\le 15$ & $1.02799438139 \times 10^{-9}$ & $-30$ & 2 & 2 \\
    \bottomrule
    \end{tabular}
    \end{center}
\end{table}

Although the fast-path budget allows three limbs for binary64, our implementation keeps four limbs and meets
the worst-case budget $\varepsilon = \ulp{w^*}$. The same routine then serves both as the complete range
reduction of a single-stage implementation and as the fast path of a correctly rounded one
(Remark~\ref{rem:fastpath}).

\section{Branch-Free Floating-Point Reduction Steps} \lab{sec:PH}

For the remainder of this paper, we focus on double precision ($p = 53$) using $L = 4$ limbs of
precision $p_c \le p$ (with $p_c = 51$ in our implementation):
\eq{}{ C_0 + C_1 + C_2 + C_3 \approx \frac{2^N}{\pi} \bmod 2^{N + 1}.}

The product $x \cdot C_i$ is represented exactly as a sum of high and low floating-point parts:
\eq{}{x \cdot C_i = (x \cdot C_i)_h + (x \cdot C_i)_\ell,}
where
\eq{SplitBound}{
    \ufp{(x \cdot C_i)_\ell} \leq \ufp{(x \cdot C_i)_h} \cdot 2^{-p} = \frac{\ulp{(x \cdot C_i)_h}}{2}.
}

\begin{Rem} \lab{Rem:UlpLsb}
    For the low part $(x \cdot C_i)_\ell$, we bound the least significant bit by:
    \eq{}{\uls{(x \cdot C_i)_\ell} \ge \uls{x \cdot C_i} = \uls{x} \cdot \uls{C_i},}
    with equality when the low part is non-zero.
    This bound is sharper than $\ulp{(x \cdot C_i)_\ell}$ because the low part can be zero or have small
    magnitude. Conversely, for the high part, $\ulp{(x \cdot C_i)_h} = \ufp{x \cdot C_i} \cdot 2^{-p+1}$
    provides the appropriate bound.
    \RemEnd
\end{Rem}

\begin{Lem} \lab{Lem:LsbLoHi}
    The low part $(x \cdot C_i)_\ell$ and high part $(x \cdot C_{i + 1})_h$ satisfy:
    \eq{LsbLoHi}{
        \uls{(x \cdot C_i)_\ell} \geq \ulp{(x \cdot C_{i + 1})_h}.
    }
\end{Lem}
\begin{proof}
    For greedy round-to-nearest limbs of precision $p_c$ (with $\ulp_{p_c}(z) := \ufp{z} \cdot 2^{-p_c + 1}$),
    $|C_{i+1}| \le \frac{1}{2} \ulp_{p_c}(C_i) = \ufp{C_i} \cdot 2^{-p_c}$.
    When consecutive limbs have precision gap $\ge p_c + 1$, $\uls{C_i} \ge 4 \ufp{C_{i+1}}$, so:
    \begin{align*}
        \uls{(x \cdot C_i)_\ell} &\ge \uls{x \cdot C_i} = \uls{x} \cdot \uls{C_i} \\
            &\geq \uls{x} \cdot 4 \ufp{C_{i + 1}} \\
            &\geq \ufp{x} \cdot 2^{-p + 1} \cdot 4 \ufp{C_{i + 1}} \\
            &\geq \ufp{(x \cdot C_{i + 1})_h} \cdot 2^{-p + 1} \\
            &= \ulp{(x \cdot C_{i + 1})_h}.
    \end{align*}
    In the boundary case where consecutive limbs have gap $p_c$ (so $C_{i+1} = \pm \frac{1}{2} \ulp_{p_c}(C_i)$
    is an exact power of two), $(x \cdot C_{i+1})_h = x \cdot C_{i+1}$ has $\ufp{(x \cdot C_{i+1})_h} = \ufp{x}
    \ufp{C_{i+1}}$. Then $\uls{C_i} \ge 2 \ufp{C_{i+1}}$, and replacing the factor of 4 with 2 yields the same
    inequality $\uls{(x \cdot C_i)_\ell} \ge \ulp{(x \cdot C_{i+1})_h}$.
\end{proof}

The alignment of the products $x \cdot C_i$ is illustrated in Figure~\ref{fig:align}.

\begin{figure}[htbp]
    \centering
    \begin{tikzpicture}[scale=0.95]
        \draw (0, 0) to node[midway, above] {$(x \cdot C_0)_h$}(2, 0);
        \draw (2, 0) to node[midway, above] {$(x \cdot C_0)_\ell$}(4, 0);
        \draw (0, -0.1) to (0, 0.1);
        \draw (2, -0.1) to (2, 0.1);
        \draw (4, -0.1) to (4, 0.1);

        \draw (2, -0.9) to node[midway, above] {$(x \cdot C_1)_h$}(4, -0.9);
        \draw (4, -0.9) to node[midway, above] {$(x \cdot C_1)_\ell$}(6, -0.9);
        \draw (2, -1.0) to (2, -0.8);
        \draw (4, -1.0) to (4, -0.8);
        \draw (6, -1.0) to (6, -0.8);

        \draw (4, -1.8) to node[midway, above] {$(x \cdot C_2)_h$}(6, -1.8);
        \draw (6, -1.8) to node[midway, above] {$(x \cdot C_2)_\ell$}(8, -1.8);
        \draw (4, -1.9) to (4, -1.7);
        \draw (6, -1.9) to (6, -1.7);
        \draw (8, -1.9) to (8, -1.7);

        \draw (6, -2.7) to node[midway, above] {$(x \cdot C_3)_h$}(8, -2.7);
        \draw (6, -2.8) to (6, -2.6);
        \draw (8, -2.8) to (8, -2.6);
\end{tikzpicture}
\caption{\lab{fig:align} Alignment of limb products $x \cdot C_i$.}
\end{figure}
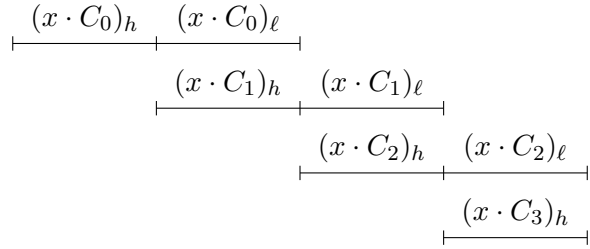

By strengthening the lower bound of~\eqr{C0Bound} slightly to:
\eq{DropxC0h}{\ufp{x} \cdot \ufp{C_0} \geq 2^{N + p},}
we guarantee that $\ulp{(x \cdot C_0)_h} \geq 2^{N + 1}$. Therefore, $(x \cdot C_0)_h$ is an integer multiple
of $2^{N+1}$, meaning $(x \cdot C_0)_h \equiv 0 \pmod{2^{N+1}}$, so $a_h$ can be discarded.

Furthermore, by imposing:
\eq{DropxC1l}{2^{-p + 1} \leq \ulp{(x \cdot C_1)_h} \leq \frac{1}{4},}
the binary point is strictly contained within $(x \cdot C_0)_\ell$ and $(x \cdot C_1)_h$. This allows
computing $k$ and the high part $v$ of the residual in four branch-free floating-point steps
(Algorithm~\ref{alg:w}).

\begin{algorithm}[H]
    \caption{First Steps of Range Reduction} \lab{alg:w}
    \begin{algorithmic}[1]
        \STATE $a_h + a_\ell \leftarrow x \cdot C_0$ \COMMENT{$a_h \equiv 0 \pmod{2^{N+1}}$ is unused}
        \STATE $b_h + b_\ell \leftarrow x \cdot C_1$
        \STATE $\hat{k} \leftarrow \nearestint{a_\ell \oplus b_h}$
        \STATE $v \leftarrow \prt{a_\ell \ominus \hat{k}} \oplus b_h$
    \end{algorithmic}
\end{algorithm}

\begin{Thm} \lab{thm:FirstSteps}
    Assuming conditions~\eqr{DropxC0h},~\eqr{DropxC1l}, and $\ufp{x} \cdot \ufp{C_0} \le 2^{p+p_c-4}$, the outputs of
    Algorithm~\ref{alg:w} satisfy:
    \begin{align}
        \abs{(a_\ell + b_h) - \hat{k}} &\leq 0.5 + \max(\ulp{a_\ell}, \ulp{b_h}) \le \frac{3}{4}, \label{eq:kBounds1} \\
        x \cdot (C_0 + C_1) &\equiv \hat{k} + v + b_\ell \pmod{2^{N + 1}}. \label{eq:kBounds2}
    \end{align}
    Furthermore, with step~3 rounded to nearest, both floating-point operations in step~4 of Algorithm~\ref{alg:w}
    are exact in $\F$ across all IEEE 754 rounding modes, and $|v| \le 3/4 < 1$ (see
    Section~\ref{subsec:dir_round} for directed rounding in step~3).
\end{Thm}

\begin{proof}
    To prove~\eqr{kBounds1}, recall from Lemma~\ref{Lem:UfpUlsProp}(c) that:
    \eq{}{
        \ufp{a_h} \leq 2 \ufp{x} \cdot \ufp{C_0}.
    }
    The unit in the last place of $a_h$ is therefore bounded by:
    \eq{}{
        \ulp{a_h} = \ufp{a_h} \cdot 2^{-p+1} \leq \ufp{x} \cdot \ufp{C_0} \cdot 2^{-p + 2}.
    }
    From~\eqr{SplitBound}, the remainder $a_\ell$ satisfies $\ufp{a_\ell} \leq \frac{1}{2} \ulp{a_h} \leq \ufp{x} \cdot
    \ufp{C_0} \cdot 2^{-p + 1}$, which implies:
    \eq{}{
    \begin{aligned}
        \ulp{a_\ell} &= \ufp{a_\ell} \cdot 2^{-p + 1} \\
            &\leq \ufp{x} \cdot \ufp{C_0} \cdot 2^{-2p + 2}.
    \end{aligned}
    }
    Under the assumption $\ufp{x} \cdot \ufp{C_0} \le 2^{p+p_c-4}$ with $p_c \le p$, we have:
    \eq{}{
        \ulp{a_\ell} \le 2^{p + p_c - 4} \cdot 2^{-2p + 2} = 2^{-p + p_c - 2} \le 2^{-2} = \frac{1}{4}.
    }
    Condition~\eqr{DropxC1l} provides:
    \eq{}{
        \ulp{b_h} \le \frac{1}{4}.
    }
    By Lemma~\ref{Lem:UfpUlsProp}(b), the rounding error of $a_\ell \oplus b_h$ under round-to-nearest
    ($\operatorname{RN}$) satisfies:
    \eq{}{
    \begin{aligned}
        \abs{(a_\ell + b_h) - (a_\ell \oplus b_h)} &\leq 2^{-p} \ufp{a_\ell + b_h} \\
            &\le 2^{-p + 1} \max\prt{\ufp{a_\ell}, \ufp{b_h}} \\
            &= \max\prt{\ulp{a_\ell}, \ulp{b_h}} \le \frac{1}{4}.
    \end{aligned}
    }
    Since $\hat{k} = \nearestint{a_\ell \oplus b_h}$ rounds to the nearest integer, we have:
    \eq{}{
        \abs{(a_\ell \oplus b_h) - \hat{k}} \leq 0.5.
    }
    Applying the triangle inequality yields~\eqr{kBounds1}:
    \eq{}{
    \begin{aligned}
        \abs{(a_\ell + b_h) - \hat{k}} &\le \abs{(a_\ell + b_h) - (a_\ell \oplus b_h)} + \abs{(a_\ell \oplus b_h) -
        \hat{k}} \\
            &\le 0.5 + \max\prt{\ulp{a_\ell}, \ulp{b_h}} \\
            &\le 0.5 + \frac{1}{4} = \frac{3}{4} < 1.
    \end{aligned}
    }

    Next, we show that the subtraction $a_\ell \ominus \hat{k}$ in step~4 of Algorithm~\ref{alg:w} is exact across all
    IEEE 754 rounding modes. Let:
    \eq{}{
        \eta := \min\prt{\uls{a_\ell}, 1}.
    }
    Because $\hat{k} \in \Z$, its least significant bit satisfies $\uls{\hat{k}} \ge 1 \ge \eta$. By
    Lemma~\ref{Lem:UfpUlsProp}(b), $a_\ell - \hat{k}$ is an integer multiple of $\eta$, so:
    \eq{}{
        \uls{a_\ell - \hat{k}} \ge \eta.
    }
    To prove that $a_\ell - \hat{k} \in \F$, by Lemma~\ref{Lem:UfpUlsProp}(a) it suffices to show that:
    \eq{}{
        \abs{a_\ell - \hat{k}} < 2^p \eta.
    }
    Under greedy round-to-nearest limb generation, the limb $C_1$ satisfies:
    \eq{}{
        \abs{C_1} \le \frac{1}{2} \uls{C_0}.
    }
    Using $|x| \le (2^p - 1)\uls{x}$ and the multiplicativity property $\uls{x \cdot C_0} = \uls{x} \cdot \uls{C_0}$
    from Lemma~\ref{Lem:UfpUlsProp}(c), whenever $a_\ell \ne 0$ we have $\uls{a_\ell} = \uls{x \cdot C_0}$
    (if $a_\ell = 0$, then $\uls{a_\ell} = +\infty > 1$, which falls directly under Case~2 below). Therefore:
    \eq{}{
    \begin{aligned}
        \abs{x \cdot C_1} &\le \abs{x} \cdot \abs{C_1} \\
            &\le (2^p - 1)\uls{x} \cdot \frac{1}{2}\uls{C_0} \\
            &= \prt{2^{p-1} - \frac{1}{2}} \uls{a_\ell}.
    \end{aligned}
    }
    Since $\prt{2^{p-1} - 1/2}\uls{a_\ell}$ is a floating-point number in $\F$, rounding preserves this bound:
    \eq{}{
        \abs{b_h} = \abs{\circ(x \cdot C_1)} \le \prt{2^{p-1} - \frac{1}{2}} \uls{a_\ell}.
    }
    Writing $a_\ell - \hat{k} = -b_h + ((a_\ell + b_h) - \hat{k})$ and applying~\eqr{kBounds1}, we have:
    \eq{}{
        \abs{a_\ell - \hat{k}} \le \abs{b_h} + \abs{(a_\ell + b_h) - \hat{k}} \le \abs{b_h} + \frac{3}{4}.
    }
    We now consider two cases:
    \begin{itemize}
        \item \textbf{Case 1: $\uls{a_\ell} \le 1$.} In this case, $\eta = \uls{a_\ell}$. We have:
        \eq{}{
            \abs{b_h} \le \prt{2^{p-1} - \frac{1}{2}} \eta.
        }
        Combining this with~\eqr{kBounds1} yields:
        \eq{}{
            \abs{a_\ell - \hat{k}} \le \prt{2^{p-1} - \frac{1}{2}} \eta + \frac{3}{4} < 2^p \eta,
        }
        where the strict inequality holds because $\eta = \uls{a_\ell} \ge \ulp{b_h} \ge 2^{-p+1}$ by
        Lemma~\ref{Lem:LsbLoHi} and~\eqr{DropxC1l}.

        \item \textbf{Case 2: $\uls{a_\ell} > 1$.} In this case, $\eta = 1$. By condition~\eqr{DropxC1l}, $\ulp{b_h} \le
            1/4$, which implies:
        \eq{}{
            \abs{b_h} < 2^p \cdot \ulp{b_h} \le 2^{p-2}.
        }
        Therefore:
        \eq{}{
            \abs{a_\ell - \hat{k}} \le \abs{b_h} + \frac{3}{4} < 2^{p-2} + \frac{3}{4} < 2^p = 2^p \eta.
        }
    \end{itemize}
    In both cases, we obtain:
    \eq{}{
        \abs{a_\ell - \hat{k}} < 2^p \eta \le 2^p \uls{a_\ell - \hat{k}}.
    }
    By Lemma~\ref{Lem:UfpUlsProp}(a), $a_\ell - \hat{k} \in \F$.
    Consequently, the subtraction introduces no rounding error across all IEEE 754 rounding modes:
    \eq{}{
        a_\ell \ominus \hat{k} = a_\ell - \hat{k}.
    }

    Next, we consider the addition in step~4:
    \eq{}{
        v = (a_\ell \ominus \hat{k}) \oplus b_h.
    }
    By Lemma~\ref{Lem:LsbLoHi}, $\uls{a_\ell} \ge \ulp{b_h}$. Condition~\eqr{DropxC1l} guarantees that:
    \eq{}{
        2^{-p+1} \le \ulp{b_h} \le \frac{1}{4} < 1.
    }
    Since $\hat{k} \in \Z$, its least significant bit satisfies $\uls{\hat{k}} \ge 1 > \ulp{b_h}$.
    Therefore, both $a_\ell$ and $\hat{k}$ are integer multiples of $\ulp{b_h}$, which implies:
    \eq{}{
        \uls{a_\ell - \hat{k}} \ge \ulp{b_h}.
    }
    Since $b_h \in \F$ is also an integer multiple of $\ulp{b_h}$, the exact sum:
    \eq{}{
        (a_\ell - \hat{k}) + b_h = (a_\ell + b_h) - \hat{k}
    }
    is an integer multiple of $\ulp{b_h}$.
    By~\eqr{kBounds1}, the magnitude of this sum is bounded by:
    \eq{}{
        \abs{(a_\ell - \hat{k}) + b_h} = \abs{(a_\ell + b_h) - \hat{k}} \le \frac{3}{4} < 1.
    }
    If $(a_\ell - \hat{k}) + b_h = 0$, the addition is trivially exact.
    If $(a_\ell - \hat{k}) + b_h \ne 0$, its most significant bit satisfies:
    \eq{}{
        \ufp{(a_\ell - \hat{k}) + b_h} \le 2^{-1}.
    }
    Since the sum is a non-zero integer multiple of $\ulp{b_h} \ge 2^{-p+1}$, its least significant bit satisfies:
    \eq{}{
        \uls{(a_\ell - \hat{k}) + b_h} \ge \ulp{b_h} \ge 2^{-p+1}.
    }
    Thus, the number of bits required to represent this non-zero sum is at most:
    \eq{}{
    \begin{aligned}
        \log_2\prt{\ufp{(a_\ell - \hat{k}) + b_h}} &- \log_2\prt{\uls{(a_\ell - \hat{k}) + b_h}} + 1 \\
            &\le (-1) - (-p + 1) + 1 = p - 1 < p.
    \end{aligned}
    }
    By Lemma~\ref{Lem:UfpUlsProp}(a), $(a_\ell - \hat{k}) + b_h \in \F$.
    Therefore, the floating-point addition introduces no rounding error across all IEEE 754 rounding modes:
    \eq{}{
        v = (a_\ell \ominus \hat{k}) \oplus b_h = (a_\ell - \hat{k}) + b_h = (a_\ell + b_h) - \hat{k},
    }
    with $\abs{v} \le 3/4 < 1$.

    Finally, we verify the equivalence~\eqr{kBounds2}. Expanding the product of $x$ with the two limbs $C_0$ and $C_1$:
    \eq{}{
        x \cdot (C_0 + C_1) = a_h + a_\ell + b_h + b_\ell.
    }
    By condition~\eqr{DropxC0h}, $\ulp{a_h} \geq 2^{N+1}$, so $a_h$ is an integer multiple of $2^{N+1}$:
    \eq{}{
        a_h \equiv 0 \pmod{2^{N+1}}.
    }
    Substituting $a_\ell + b_h = \hat{k} + v$ into the expansion, we obtain:
    \eq{}{
    \begin{aligned}
        x \cdot (C_0 + C_1) &= a_h + (a_\ell + b_h) + b_\ell \\
            &\equiv \hat{k} + v + b_\ell \pmod{2^{N+1}},
    \end{aligned}
    }
    which completes the proof.
\end{proof}

Combining~\eqr{DropxC0h},~\eqr{DropxC1l}, and~\eqr{kBounds1}, the parameters must satisfy:
\eq{msbBounds}{
    \begin{cases}
        2^{N + p} \leq \ufp{x} \cdot \ufp{C_0} \leq 2^{p + p_c - 4}, \\
        \ufp{x} \cdot \ufp{C_1} \geq 1.
    \end{cases}
}

\begin{Rem}
    Enforcing the tighter bound $|v| \leq 0.5$ would require conditional branching to adjust $\hat{k}
    \leftarrow \hat{k} \pm 1$. Allowing $|v| < 1$ keeps the algorithm branch-free, and the evaluation polynomial
    handles the slightly wider range without loss of accuracy.
    \RemEnd
\end{Rem}

\section{Fast Lookup Table from Exponent Field} \lab{sec:LUT}

Condition~\eqr{msbBounds} provides a valid interval for $\ufp{x}$ given a quadruple $(C_0, C_1, C_2,
C_3)$. We exploit this flexibility by sharing the same quadruple across $M$ consecutive exponents of $x$.

Let $\PiInv$ be the lookup table and $i$ be the table index corresponding to $x = (-1)^{s_x} \cdot 2^{e_x}
\cdot m_x$ with significand $m_x \in [1, 2)$ (so $\ufp{x} = 2^{e_x}$). Let $x_{\min}$ and $x_{\max}$ denote
the smallest and largest positive numbers mapped to index $i$, so $\ufp{x_{\max}} = \ufp{x_{\min}} \cdot 2^{M
- 1}$. Both must satisfy~\eqr{msbBounds}:
\eq{msbXMin}{
    \begin{aligned}
        2^{N + p} &\leq 2^{e_{x_{\min}}} \cdot \ufp{C_0} \\
                  &\leq 2^{e_{x_{\max}}} \cdot \ufp{C_0} \leq 2^{p + p_c - 4},
    \end{aligned}
}
which is equivalent to:
\eq{msbXMin2}{
    2^{N + p} \leq 2^{e_{x_{\min}}} \cdot \ufp{C_0} \leq 2^{p + p_c - M - 3},
}
implying the bound:
\eq{MBound}{M \leq p_c - N - 3.}

We map $e_x$ to index $i$ using integer division:
\eq{IdxForm}{i = \floor{\frac{e_x - d}{M}},}
where $M$ is chosen as a power of 2 so that division becomes an arithmetic right shift.
To define entry $i$, we take:
\eq{Ci}{C_0 + C_1 + C_2 + C_3 \approx 2^{N - M i} \cdot \cbrk{\frac{2^{M i}}{\pi}},}
where $\cbrk{a} = a - \nearestint{a} \in [-0.5, 0.5]$ denotes the centered fractional part. Let
\eq{fg}{
    f := \min_i \left(-\log_2 \ufp{\cbrk{\frac{2^{Mi}}{\pi}}}\right), \quad
    g := \max_i \left(-\log_2 \ufp{\cbrk{\frac{2^{Mi}}{\pi}}}\right).
}
Then condition~\eqr{msbXMin2} holds if:
\eq{XMinBounds}{
    \begin{aligned}
        2^{N + p + g} &\leq 2^{e_{x_{\min}} + N - Mi} \\
                      &\leq 2^{p + p_c - M + f - 3}.
    \end{aligned}
}
From the index formula~\eqr{IdxForm}, the boundary exponents mapped to entry $i$ are:
\eq{ExBoundaries}{
    e_{x_{\min}} = Mi + d, \qquad e_{x_{\max}} = Mi + d + M - 1.
}
Substituting $e_{x_{\min}} - Mi = d$ into the lower bound of~\eqr{XMinBounds} yields:
\eq{dLowerBound}{
    2^{N + p + g} \le 2^{e_{x_{\min}} + N - Mi} = 2^{N + d},
}
which requires:
\eq{}{
    d \ge p + g.
}
We now determine the value of $d$ that minimizes the truncation error.
From~\eqr{TruncErr}, the truncation error modulo $2^{N+1}$ from omitting limbs beyond $C_{L-1}$ for any input
$x$ mapped to index $i$ is bounded by:
\eq{TruncErrD}{
\begin{aligned}
    \abs{x \cdot \sum_{j=L}^\infty C_j \bmod 2^{N+1}} &< 4 \ufp{x} \cdot \ufp{C_L} \\
        &\le 4 \cdot 2^{e_{x_{\max}}} \cdot \ufp{C_0} \cdot 2^{-L p_c} \\
        &\le 4 \cdot 2^{Mi + d + M - 1} \cdot 2^{N - Mi - f} \cdot 2^{-L p_c} \\
        &= 2^{N + d + M - f + 1 - L p_c}.
\end{aligned}
}
Because this upper bound increases strictly monotonically with $d$, minimizing the error requires choosing
the smallest valid offset $d = p + g$.
This yields the optimal table index formula:
\eq{idx}{
    i = \floor{\frac{e_x - p - g}{M}},
}
and bounds the truncation error across all binades mapped to index $i$ by:
\eq{TruncErrBound}{
    \abs{x \cdot \sum_{j=L}^\infty C_j \bmod 2^{N+1}} < 2^{N + p + g + M - f + 1 - L p_c}.
}
Substituting $e_{x_{\min}} - Mi = p + g$ into the upper bound of~\eqr{XMinBounds} yields:
\eq{}{
    (p + g) + N \le p + p_c - M + f - 3,
}
which is equivalent to:
\eq{MBound1}{
    M \le p_c + f - 3 - N - g.
}

Finally, requiring the truncation error~\eqr{TruncErrBound} to stay within the budget $\varepsilon$ of
Section~\ref{sec:WorstCase} gives:
\eq{MBoundEps}{
    M \le L p_c + \log_2 \varepsilon - N - p - g + f - 1.
}
Compared with~\eqr{LimbBounds2}, sharing one entry across $M$ binades costs $M + g - f$ bits of the budget.

For our double-precision table, the parameters are $p = 53$, $p_c = 51$, $N = 7$, $g = 9$, and $f = 2$, where $f$
is attained at $i = 0$ since $1/\pi \in [1/4, 1/2)$. Condition~\eqr{MBound1} gives:
\eq{}{
    M \le 51 + 2 - 3 - 7 - 9 = 34,
}
and condition~\eqr{MBoundEps} gives:
\begin{itemize}
    \item For the worst-case budget $\varepsilon = \ulp{w^*} = 2^{-114}$ with $L = 4$:
    \eq{}{
        M \le 4 \cdot 51 - 114 - 7 - 53 - 9 + 2 - 1 = 22,
    }
    so $M = 16$ is the largest admissible power of two. With $M = 16$, the omitted tail over all entries is
    below $2^{-126}$ (proof of Corollary~\ref{cor:llvm_params}).
    \item For the fast-path budget $\varepsilon = 2^{-67}$ with $L = 3$:
    \eq{}{
        M \le 3 \cdot 51 - 67 - 7 - 53 - 9 + 2 - 1 = 18,
    }
    so the same table without $D_3$ meets the fast-path budget with 2 bits to spare: $M = 16$ remains
    admissible down to $\varepsilon = 2^{-69}$, whereas $\varepsilon = 2^{-70}$ would require $M \le 15$. The
    actual omitted tail over all entries is:
    \eq{}{
        \abs{x_r \sum_{j=3}^\infty D_j} < 2^{-72}.
    }
    \item For the fast-path budget $\varepsilon = 2^{-67}$ with $L = 4$:
    \eq{}{
        M \le 4 \cdot 51 - 67 - 7 - 53 - 9 + 2 - 1 = 69,
    }
    so only~\eqr{MBound1} limits $M$. For example, $M = 32$ halves the table (Table~\ref{tab:mem}) at the cost
    of a looser quotient estimate $\hat{k}$.
\end{itemize}
Condition~\eqr{MBoundEps} is sufficient but not necessary: the $M = 32$ configuration with 53-bit limbs in
Table~\ref{tab:mem} satisfies only $M \le 30$ (with the same $f$ and $g$), and is certified instead by
evaluating the omitted tails of all 33 entries exactly.

For the single-precision table, the parameters are $p = p_c = 24$, $N = 3$, $g = 8$, and $f = 2$. For the
worst-case budget $\varepsilon = \ulp{w^*} = 2^{-53}$ with $L = 4$, condition~\eqr{MBoundEps} gives:
\eq{}{
    M \le 4 \cdot 24 - 53 - 3 - 24 - 8 + 2 - 1 = 9,
}
hence $M = 8$. Three limbs are not enough even for the fast-path budget $\varepsilon = 2^{-38}$, since:
\eq{}{
    M \le 3 \cdot 24 - 38 - 3 - 24 - 8 + 2 - 1 = 0,
}
and the actual three-limb tail exceeds $2^{-34} > 2^{-38}$, so single precision keeps four limbs. Although
Table~\ref{tab:limbs} allows three limbs for $N \le 9$, their slack in~\eqr{LimbBounds2} for $N = 3$ is only:
\eq{}{
    3 \cdot 24 - (N + p + 1 - \log_2 \varepsilon) = 72 - (3 + 24 + 1 + 38) = 6
}
bits, while sharing one entry across $M = 8$ binades costs $M + g - f = 14$ bits.

\textbf{Scaled Constants for Underflow Avoidance:}
Directly storing $C_0, \dots, C_3$ from~\eqr{Ci} would cause exponent underflow for large $i$. Instead, we store
pre-scaled constants:
\eq{Di}{
\begin{aligned}
    D_0 + D_1 + D_2 + D_3 &= 2^{M i} (C_0 + C_1 + C_2 + C_3) \\
                          &\approx 2^N \cbrk{\frac{2^{M i}}{\pi}},
\end{aligned}
}
and scale the input using fast exponent adjustment:
\eq{Xr}{x_r = x \cdot 2^{-M i}.}
Then:
\eq{ScaledProd}{
\begin{aligned}
    x_r (D_0 + D_1 + D_2 + D_3) &= x (C_0 + C_1 + C_2 + C_3) \\
                                &\approx x \cdot \frac{2^N}{\pi} \bmod 2^{N + 1}.
\end{aligned}
}

\textbf{Input Domain Extension and Exponent Lower Bounds:}
We consider the domain extension and exponent lower bounds for general parameters $p$, $p_c$, $N$, $M$, $g$,
and prefix exponent $j \in \Z^+$.

\emph{1. Baseline Exponent Bound ($i \ge 0$):}
For $i \ge 0$, formula~\eqr{idx} requires:
\eq{BaseExpBound}{
    e_x \ge p + g.
}
Under this condition, $\ulp{(x_r \cdot D_0)_h} \ge 2^{N+1}$, so $(x_r \cdot D_0)_h \equiv 0 \pmod{2^{N+1}}$ can
be discarded by condition~\eqr{DropxC0h}.
When $e_x < p + g$, the table index $i$ is negative.
Because $2^{Mi}/\pi < 1$ for all $i < 0$, the centered fractional part satisfies:
\eq{}{
    \ufp{\cbrk{\frac{2^{Mi}}{\pi}}} = \ufp{\frac{2^{Mi}}{\pi}} = 2^{Mi - 2}.
}
As $i$ decreases, $\ufp{D_0}$ decreases exponentially, violating condition~\eqr{DropxC0h}.

\emph{2. Prefix Extension and Modulus Invariance:}
To extend the algorithm to negative indices without runtime branching, we modify $D_0$ for $i < 0$ by adding a
constant prefix $2^{-j}$ to $\{2^{Mi}/\pi\}$:
\eq{DiPrefixGen}{
    D_0 \approx 2^N \prt{2^{-j} + \frac{2^{Mi}}{\pi}} = 2^{N-j} + 2^N \cdot \frac{2^{Mi}}{\pi}.
}
The prefix adds $2^{N-j} x_r$ to the leading product.
Because $e_{x_r} = e_x - Mi \ge p + g$, we have:
\eq{}{
    \uls{x_r} \ge 2^{g + 1}.
}
For the prefix contribution to vanish modulo $2^{N+1}$, it suffices that:
\eq{}{
    \uls{2^{N-j} x_r} = 2^{N-j} \uls{x_r} \ge 2^{N+1},
}
which holds if and only if:
\eq{}{
    j \le g.
}
In this case, $2^{N-j} x_r \equiv 0 \pmod{2^{N+1}}$ introduces no error.
Furthermore, the prefix sets $\ufp{D_0} = 2^{N-j}$, ensuring:
\eq{}{
    \ulp{(x_r \cdot D_0)_h} \ge 2^{N + g - j + 1} \ge 2^{N+1}.
}

\emph{3. Exponent Lower Bound:}
For $2 \le j \le g$, $\ufp{D_0} = 2^{N-j}$ lies between $2^{N-g}$ (item~2) and the value $2^{N-2}$ of the
$i = 0$ entry, so the first line of~\eqr{msbBounds} holds for $i < 0$ as it does for $i \ge 0$. The
binding constraint is the second line, $\ufp{x_r} \cdot \ufp{D_1} \ge 1$. Let
\eq{BitDistGen}{
    \delta(i) := (N - j) - (N + Mi - 2) = -Mi + 2 - j
}
be the distance between the leading bits of the prefix and of the tail $2^N \cdot \frac{2^{Mi}}{\pi}$. For
$\delta(i) \le p_c$, $D_1$ is an ordinary second limb and is checked entry by entry, as for $i \ge 0$.
For $\delta(i) \ge p_c + 1$, rounding gives $D_0 = 2^{N-j}$ exactly and the tail moves to $D_1$, which is
$2^{N+Mi}/\pi$ rounded to $p_c$ bits. Its significand is that of $4/\pi$, so $\ufp{D_1} = 2^{N+Mi-2}$ and,
since $\ufp{x_r} \ge 2^{p+g}$,
\eq{iMinGen}{
    \ufp{x_r} \cdot \ufp{D_1} \ge 2^{p + g + N + Mi - 2} \ge 1
    \iff i \ge i_{\min} := -\floor{\frac{p + g + N - 2}{M}}.
}
This yields the lower bound on input exponents:
\eq{ExMinGen}{
    e_x \ge e_{x, \min} := p + g + M \cdot i_{\min}.
}
The truncation criterion is checked identically; for these negative-index entries it is dominated by the $i \ge 0$
entries.

\emph{4. Application to Double Precision ($p=53$):}
For binary64, the production parameters are $p = 53, p_c = 51, N = 7, M = 16, g = 9$, and $j = 2$ (prefix
$2^{-2} = 0.25$, scaling to $2^{N-2} = 32$):
\begin{itemize}
    \item \textbf{Baseline bound:} $e_x \ge p + g = 53 + 9 = 62$ ($|x| \ge 2^{62}$).
    \item \textbf{Prefix condition:} $j = 2 \le g = 9$, ensuring $32 x_r \equiv 0 \pmod{256}$ since $\uls{x_r}
        \ge 1024$.
    \item \textbf{Exponent lower bound:} With $p = 53, g = 9, N = 7, M = 16$, condition~\eqr{iMinGen} gives:
    \eq{}{
        i_{\min} = -\floor{\frac{67}{16}} = -4,
    }
    so $e_{x, \min} = 62 + 16 \cdot (-4) = -2$. At $i = -4$, $D_0 = 32$ exactly and:
    \eq{}{
        \ufp{x_r} \cdot \ufp{D_1} \ge 2^{53 + 9 + 7 - 64 - 2} = 2^3 \ge 1,
    }
    whereas at $i = -5$ the same bound is $2^{-13} < 1$.
    \item \textbf{Table footprint and domain:} Covering up to $e_{x, \max} = 1023$ gives:
    \eq{}{
        i_{\max} = \floor{\frac{1023 - 62}{16}} = 60.
    }
    The resulting table spans $i \in [-3, 60]$ (64 entries, 2.0~KB). In production, inputs with $|x| < 2^{16}$
    ($e_x < 16$) use Cody--Waite range reduction where $k < 2^{22}$, so indices $i < -3$ are not needed.
\end{itemize}

\emph{5. Application to Single Precision ($p=24$):}
For binary32, the parameters are $p = 24, p_c = 24, N = 3, M = 8, g = 8$, and $j = 2$ (prefix
$2^{-2} = 0.25$, scaling to $2^{N-2} = 2$):
\begin{itemize}
    \item \textbf{Baseline bound:} $e_x \ge p + g = 24 + 8 = 32$ ($|x| \ge 2^{32}$).
    \item \textbf{Prefix condition:} $j = 2 \le g = 8$, ensuring $2 x_r \equiv 0 \pmod{16}$ since $\uls{x_r} \ge
        512$.
    \item \textbf{Exponent lower bound:} With $p = 24, g = 8, N = 3, M = 8$, condition~\eqr{iMinGen} gives:
    \eq{}{
        i_{\min} = -\floor{\frac{33}{8}} = -4,
    }
    so $e_{x, \min} = 32 + 8 \cdot (-4) = 0$. At $i = -4$, $D_0 = 2$ exactly and:
    \eq{}{
        \ufp{x_r} \cdot \ufp{D_1} \ge 2^{24 + 8 + 3 - 32 - 2} = 2 \ge 1.
    }
    \item \textbf{Table footprint and domain:} Covering up to $e_{x, \max} = 127$ gives:
    \eq{}{
        i_{\max} = \floor{\frac{127 - 32}{8}} = 11.
    }
    The resulting table spans $i \in [-2, 11]$, requiring 14 entries (0.22~KB, as listed in
    Table~\ref{tab:mem}). For $|x| < 2^{18}$, Cody--Waite range reduction is used since
    $k < 2^{18} \cdot 8/\pi \approx 6.7 \times 10^5 < 2^{20}$, so indices $i < -2$ are not needed.
\end{itemize}

\begin{table}[htbp]
    \caption{Step Sizes and Memory Footprint} \lab{tab:mem}
    \begin{center}
    \begin{tabular}{cccc}
    \toprule
        $p$ & $M$ & Table Size & Memory Usage \\
    \midrule
        24 & 8 & 14 entries & 0.22 KB \\
        53 & 16 & 64 entries & 2.0 KB \\
        53 & 32$^*$ & 33 entries & 1.03 KB \\
        \bottomrule
    \end{tabular}
    \end{center}
    \vspace{1ex}
    \raggedright{\footnotesize For $p=24, M=8$, 14 entries cover $|x| \ge 2^{18}$. For $p=53$, $M=16$ provides
    12 bits of margin below the worst-case budget $\varepsilon = \ulp{w^*}$ with 51-bit limbs; the hypothetical
    $M=32$ configuration ($*$, covering $e_x \ge 16$ via $i \in [-2, 30]$ with 33 entries) has a 3-bit truncation
    margin and requires full 53-bit limbs with hardware FMA, since Proposition~\ref{prop:dekker_exact} requires
    $p_c \le 52$ on FMA-free targets.}
\end{table}

\textbf{Memory Footprint and Cache Considerations:}
Table~\ref{tab:mem} shows the table sizes for different choices of $M$. For double precision ($p=53$),
choosing $M=16$ yields 64 entries of four doubles each, totaling 2~KB. On modern processors where the L1 data
cache is typically 32--48~KB per core, a 2~KB table takes up only a small fraction (around 4--6\%) of the
cache and generally stays cache-resident. In comparison, a table indexing each exponent
individually (such as in SLEEF~\cite{SLEEF}, with 969 entries totaling $\approx 30$~KB) can take up most of the L1
cache. Grouping exponents with $M=16$ keeps the table small while avoiding multi-word integer arithmetic.

\begin{Rem} \lab{rem:integer_table}
    The idea of grouping exponents to index a small lookup table and choosing the limb alignments carefully
    also applies to integer-only implementations. In particular, it enables Payne--Hanek range reduction using
    a compact, non-duplicated table of $2/\pi$ with no overlapping entries across
    the entire double-precision range. We plan to present the detailed mathematical formulation, limb
    alignments, and performance of this integer fixed-point algorithm in a follow-up paper.
    \RemEnd
\end{Rem}

\section{The Complete Algorithm and Correctness Proof} \lab{sec:Main}

We summarize the complete range reduction routine in Algorithm~\ref{alg:main}, and analyze its correctness and
error bounds.

\begin{algorithm}[t]
    \caption{Branch-Free Floating-Point Range Reduction} \lab{alg:main}
    \begin{algorithmic}[1]
        \REQUIRE $x \in \F$ with $|x| \ge 2^{e_{x,\min}}$ (e.g., $|x| \ge 2^{16}$ for double precision)
        \ENSURE $\hat{k} \bmod 2^{N + 1}$, and double-double $\hat{u} = u_h + u_\ell$ (or $\hat{y} = y_h + y_\ell$)
        \STATE Extract unbiased exponent $e_x \ge e_{x,\min}$ from $x = (-1)^{s_x} 2^{e_x} m_x$
        \STATE $i \leftarrow \floor{\frac{e_x - p - g}{M}}$ \COMMENT{Table index $i \in [i_{\min}, i_{\max}]$}
        \STATE Scale input $x_r \leftarrow x \cdot 2^{-M i}$ via exponent adjustment
        \STATE $D_0, D_1, D_2, D_3 \leftarrow \operatorname{PI\_INV}[i]$
        \STATE $a_h + a_\ell \leftarrow x_r \cdot D_0$ \COMMENT{Exact product, $\ulp{a_h} \ge 2^{N+1}$}
        \STATE $b_h + b_\ell \leftarrow x_r \cdot D_1$ \COMMENT{Exact product, $\ulp{b_h} \le 1/4$}
        \STATE $c_h + c_\ell \leftarrow x_r \cdot D_2$ \COMMENT{Exact product}
        \STATE $\hat{k} \leftarrow \nearestint{a_\ell \oplus b_h}$ \COMMENT{Integer quotient}
        \STATE $v \leftarrow (a_\ell \ominus \hat{k}) \oplus b_h$ \COMMENT{Exact operations by
            Theorem~\ref{thm:FirstSteps}}
        \STATE $q_h + q_\ell \leftarrow \fts{b_\ell, c_h}$
            \COMMENT{Exact sum since $\uls{b_\ell} \ge \ulp{c_h}$~\cite{JZ25}}
        \STATE $r \leftarrow \circ(c_\ell + x_r \cdot D_3)$ \COMMENT{Fused multiply-add tail evaluation}
        \STATE $u_h + t_\ell \leftarrow \fts{v, q_h}$ \COMMENT{Exact sum since $\uls{v} \ge \ulp{q_h}$~\cite{JZ25}}
        \STATE $u_\ell \leftarrow t_\ell \oplus (q_\ell \oplus r)$
        \STATE $y_h + y_\ell \leftarrow (u_h + u_\ell) \cdot \frac{\pi}{2^N}$ \lab{alg:step14}
            \COMMENT{Optional scaled output}
        \RETURN $\hat{k} \bmod 2^{N+1}$, $\hat{u} = u_h + u_\ell$ (or $\hat{y} = y_h + y_\ell$)
    \end{algorithmic}
\end{algorithm}

In LLVM libc~\cite{L}, the large range reduction routine performs step~\ref{alg:step14} to produce the
reduced argument $\hat{y} = y_h + y_\ell$. Here $\frac{\pi}{2^N}$ is represented as a precomputed
double-double constant $P_h + P_\ell \approx \frac{\pi}{2^N}$ (\texttt{PI\_OVER\_128\_DD} in LLVM libc),
and the product $(u_h + u_\ell) \cdot (P_h + P_\ell)$ is evaluated via double-double multiplication.
Alternatively, when subsequent polynomial evaluation is formulated directly for $\sin(\pi \cdot)$ and $\cos(\pi
\cdot)$ (as in $\sinpi/\cospi$-style evaluations), step~\ref{alg:step14} can be omitted, using the normalized residual
$\hat{u} = u_h + u_\ell$ directly and saving the double-double multiplication.

\begin{Rem} \lab{rem:params}
    In the LLVM libc project~\cite{L}, Algorithm~\ref{alg:main} is implemented with parameters: single
    precision ($p=24, N=3, M=8, g=8$) and double precision ($p=53, N=7, M=16, g=9$).
    \RemEnd
\end{Rem}

We summarize the correctness and error bounds of Algorithm~\ref{alg:main} for general parameters in the
following theorem:

\begin{Thm}[Correctness and General Error Bounds] \lab{thm:correctness}
    Let $p \ge 2$ be the floating-point precision, $p_c \le p$ the table limb precision, $N \in \Z^+$ the
    modulus exponent, $M \in \Z^+$ the exponent block size, and $g \in \Z^+$ the leading-zero bound on
    $\{2^{Mi}/\pi\}$ from~\eqr{fg}.
    Let $x \in \F$ be a normal floating-point number with $|x| \ge 2^{e_{x,\min}}$, where $e_{x,\min} = p + g + M \cdot
    i_{\min}$, and let $k = \nearestint{x \cdot \frac{2^N}{\pi}}$.
    Assume the table limbs $D_0, D_1, D_2, D_3 \in \F$ satisfy the structural conditions of Section~\ref{sec:LUT},
    namely~\eqr{msbBounds} (i.e., $2^{N+p} \le \ufp{x_r}\ufp{D_0} \le 2^{p+p_c-4}$ and $\ufp{x_r}\ufp{D_1} \ge 1$),
    $\ulp{b_h} \le 1/4$, $\uls{b_\ell} \ge \ulp{c_h}$, and $p_c \ge 2$.
    Under round-to-nearest mode ($\operatorname{RN}$), Algorithm~\ref{alg:main} satisfies the following properties:
    \begin{enumerate}
        \item \textbf{Discarding the Leading Product:}
        The high part $a_h = \circ(x_r \cdot D_0)$ satisfies:
        \eq{}{
            \ulp{a_h} \ge 2^{N+1},
        }
        so $a_h \equiv 0 \pmod{2^{N+1}}$. Discarding $a_h$ introduces no error modulo $2^{N+1}$.

        \item \textbf{Quotient Bounds and Exact High Residual:}
        Let:
        \eq{DeltaKDef}{
            \Delta_k := \frac{1}{2}\ulp{a_\ell \oplus b_h} + \abs{b_\ell} + \abs{x_r} \sum_{j=2}^\infty \abs{D_j}.
        }
        If $\Delta_k < 1/2$, the integer quotient $\hat{k}$ satisfies:
        \eq{}{
            \hat{k} \in \{k-1, k, k+1\} \pmod{2^{N+1}},
        }
        with $\hat{k} \equiv k \pmod{2^{N+1}}$ whenever the exact fractional part $\left\{x \cdot
        \frac{2^N}{\pi}\right\}$ is separated from $\pm 1/2$ by more than $\Delta_k$.
        The subtraction $a_\ell \ominus \hat{k} = a_\ell - \hat{k}$ and addition $v = (a_\ell \ominus \hat{k}) \oplus
        b_h$ are exact in $\F$ across all IEEE 754 rounding modes, with:
        \eq{}{
            \abs{v} \le \frac{1}{2} + \frac{1}{2}\ulp{a_\ell \oplus b_h}, \quad \text{and} \quad
            \abs{u} \le \frac{1}{2} + \Delta_k,
        }
        where $u = \prt{x \cdot \frac{2^N}{\pi} - \hat{k}} \bmod 2^{N+1}$.

        \item \textbf{Exactness of Intermediate Fast2Sum Additions:}
        The limb products satisfy $\uls{b_\ell} \ge \ulp{c_h}$, which guarantees that $\fts{b_\ell, c_h}$ in step~10 is
        exact under round-to-nearest mode ($\operatorname{RN}$):
        \eq{}{
            q_h + q_\ell = b_\ell + c_h.
        }
        Furthermore, since $\uls{v} \ge \ulp{b_h} \ge \ulp{q_h}$, $\fts{v, q_h}$ in step~12 is also exact:
        \eq{}{
            u_h + t_\ell = v + q_h.
        }

        \item \textbf{Accuracy of the Double-Double Residual:}
        The computed double-double residual $\hat{u} = u_h + u_\ell$ satisfies:
        \eq{uBoundGen}{
            \abs{\prt{u_h + u_\ell} - \prt{x \cdot \frac{2^N}{\pi} - \hat{k}} \bmod 2^{N+1}}
                \le \frac{1}{4}\ulp{\ulp{u_h}} + \epsilon_{\text{tail}},
        }
        where:
        \eq{epsTailDef}{
            \epsilon_{\text{tail}} := \frac{1}{2} \ulp{r} + \frac{1}{2} \ulp{q_\ell \oplus r} + 2^{-p}\abs{q_\ell \oplus
            r} + \abs{x_r \sum_{j=4}^\infty D_j},
        }
        with $r = \circ(c_\ell + x_r D_3)$.
        In particular, because $|u_h| < 1$, we have $\ulp{\ulp{u_h}} \le 2^{-2p+1}$, providing the absolute bound:
        \eq{}{
            \abs{\hat{u} - u} \le 2^{-2p-1} + \epsilon_{\text{tail}}.
        }

        \item \textbf{Accuracy of the Scaled Reduced Argument:}
        Let $P = P_h + P_\ell \approx \frac{\pi}{2^N}$ be a double-double approximation with $|P - \frac{\pi}{2^N}| \le
        \epsilon_P$, and suppose the double-double multiplication in step~\ref{alg:step14} evaluates $(u_h + u_\ell)
        \cdot (P_h + P_\ell)$ with arithmetic error at most $\epsilon_{\text{mult}}$.
        Then the scaled double-double $\hat{y} = y_h + y_\ell$ satisfies:
        \eq{yBoundGen}{
            \abs{\prt{y_h + y_\ell} - \prt{x - \hat{k} \cdot \frac{\pi}{2^N}} \bmod 2\pi}
                \le \abs{\hat{u} - u} \cdot \frac{\pi}{2^N} + \epsilon_{\text{mult}} + \epsilon_P \prt{\frac{1}{2} +
                \Delta_k}.
        }
    \end{enumerate}
\end{Thm}

\begin{proof}
    \emph{Proof of (1):}
    From~\eqr{ScaledProd} in Section~\ref{sec:LUT}, exponent scaling preserves the product:
    \eq{}{
        x_r \cdot \sum_{j=0}^3 D_j = x \cdot \sum_{j=0}^3 C_j.
    }
    By condition~\eqr{DropxC0h}, we have:
    \eq{}{
        \ufp{x_r} \cdot \ufp{D_0} = \ufp{x} \cdot \ufp{C_0} \ge 2^{N+p}.
    }
    The unit in the last place of $a_h = \circ(x_r \cdot D_0)$ satisfies:
    \eq{}{
    \begin{aligned}
        \ulp{a_h} &= \ufp{a_h} \cdot 2^{-p+1} \\
            &\ge (\ufp{x_r} \cdot \ufp{D_0}) \cdot 2^{-p+1} \\
            &\ge 2^{N+p} \cdot 2^{-p+1} = 2^{N+1}.
    \end{aligned}
    }
    Because $a_h \in \F$, $a_h$ is an integer multiple of $\ulp{a_h} \ge 2^{N+1}$.
    Therefore:
    \eq{}{
        a_h \equiv 0 \pmod{2^{N+1}},
    }
    and dropping $a_h$ introduces no error modulo $2^{N+1}$.

    \emph{Proof of (2):}
    By Theorem~\ref{thm:FirstSteps} with $x_r D_j = x C_j$, the subtraction $a_\ell \ominus \hat{k} = a_\ell -
    \hat{k}$ and the addition $v = (a_\ell \ominus \hat{k}) \oplus b_h$ are exact in $\F$ across all IEEE 754 rounding
    modes. Moreover, the value of $v$ satisfies:
    \eq{}{
    \begin{aligned}
        \abs{v} &= \abs{(a_\ell + b_h) - \hat{k}} \\
            &\le \abs{(a_\ell \oplus b_h) - \hat{k}} + \abs{(a_\ell + b_h) - (a_\ell \oplus b_h)} \\
            &\le \frac{1}{2} + \frac{1}{2}\ulp{a_\ell \oplus b_h}.
    \end{aligned}
    }
    Because $\hat{k} = \nearestint{a_\ell \oplus b_h}$ omits the tail terms $b_\ell, c_h, c_\ell, \dots$,
    the difference between $a_\ell \oplus b_h$ and the exact sum modulo $2^{N+1}$ is bounded by:
    \eq{}{
    \begin{aligned}
        &\abs{(a_\ell \oplus b_h) - \prt{x_r \sum_{j=0}^\infty D_j \bmod 2^{N+1}}} \\
            &\quad \le \abs{(a_\ell \oplus b_h) - (a_\ell + b_h)} + \abs{b_\ell}
                + \abs{x_r \sum_{j=2}^\infty D_j} \\
            &\quad \le \frac{1}{2}\ulp{a_\ell \oplus b_h} + \abs{b_\ell}
                + \abs{x_r} \sum_{j=2}^\infty \abs{D_j} = \Delta_k.
    \end{aligned}
    }
    Consequently, whenever the exact fractional part $\left\{x \cdot \frac{2^N}{\pi}\right\}$ is separated from
    $\pm 1/2$ by more than $\Delta_k$, rounding $a_\ell \oplus b_h$ produces the exact nearest integer:
    \eq{}{
        \hat{k} \equiv k \pmod{2^{N+1}}.
    }
    In all cases, $\hat{k}$ differs from $k \bmod 2^{N+1}$ by at most $\pm 1$, and the exact residual satisfies:
    \eq{}{
        \abs{u} = \abs{\prt{x \cdot \frac{2^N}{\pi} - \hat{k}} \bmod 2^{N+1}} \le \frac{1}{2} + \Delta_k.
    }

    \emph{Proof of (3):}
    By Lemma~\ref{Lem:LsbLoHi}, the low part $b_\ell$ and high part $c_h$ satisfy $\uls{b_\ell} \ge \ulp{c_h}$.
    Under round-to-nearest ($\operatorname{RN}$), Jeannerod and Zimmermann~\cite[Theorem~1]{JZ25} proved that
    $\uls{a} \ge \ulp{b}$ is sufficient for $\fts{a, b}$ to be exact without requiring $|a| \ge |b|$ (see
    also~\cite{PLN26} for conditions covering all faithful rounding modes, used in Section~\ref{subsec:dir_round}).
    Therefore, the first Fast2Sum in step~10 is exact:
    \eq{}{
        q_h + q_\ell = b_\ell + c_h.
    }
    Next, we verify the precondition for the second Fast2Sum $\fts{v, q_h}$ in step~12.
    From Theorem~\ref{thm:FirstSteps}, $v = (a_\ell - \hat{k}) + b_h$ is an integer multiple of $\ulp{b_h}$, so if
    $v \ne 0$, its least significant bit satisfies:
    \eq{}{
        \uls{v} \ge \ulp{b_h}.
    }
    Meanwhile, the term $c_h = \circ(x_r \cdot D_2)$ satisfies $|c_h| < 2^{p - p_c}\ulp{b_h}$.
    Since $|b_\ell| \le \frac{1}{2}\ulp{b_h}$, we have:
    \eq{}{
        \abs{q_h} = \abs{b_\ell \oplus c_h} < \prt{2^{p - p_c} + \frac{1}{2}} \ulp{b_h}.
    }
    Because $p_c \ge 2$, the unit in the last place of $q_h$ satisfies:
    \eq{}{
        \ulp{q_h} \le 2^{-p + 1 + \ceil{\log_2(2^{p - p_c} + 1/2)}} \ulp{b_h} \le \ulp{b_h} \le \uls{v}.
    }
    By the same criterion~\cite{JZ25}, $\fts{v, q_h}$ is exact under round-to-nearest:
    \eq{}{
        u_h + t_\ell = v + q_h.
    }
    If $v = 0$, exactness holds trivially.

    \emph{Proof of (4):}
    Summing the terms of the range reduction modulo $2^{N+1}$:
    \eq{}{
        \prt{x_r \sum_{j=0}^3 D_j - \hat{k}} \bmod 2^{N+1} = (v + q_h) + q_\ell + c_\ell + x_r D_3.
    }
    Since step~12 is exact ($u_h + t_\ell = v + q_h$), substituting this relation gives:
    \eq{}{
        \prt{x_r \sum_{j=0}^\infty D_j - \hat{k}} \bmod 2^{N+1} = u_h + t_\ell + q_\ell + c_\ell + x_r D_3 + x_r
        \sum_{j=4}^\infty D_j.
    }
    The error in the computed double-double $\hat{u} = u_h + u_\ell$, where $u_\ell = t_\ell \oplus (q_\ell \oplus r)$
    and $r = \circ(c_\ell + x_r D_3)$, arises from four sources:
    \begin{enumerate}
        \item \emph{Rounding of the tail product $r = \circ(c_\ell + x_r D_3)$ in step~11:}
        The rounding error of the fused multiply-add is bounded by:
        \eq{}{
            \abs{r - (c_\ell + x_r D_3)} \le \frac{1}{2} \ulp{r}.
        }

        \item \emph{Rounding of the low addition $q_\ell \oplus r$ in step~13:}
        The rounding error is bounded by:
        \eq{}{
            \abs{(q_\ell \oplus r) - (q_\ell + r)} \le \frac{1}{2} \ulp{q_\ell \oplus r}.
        }

        \item \emph{Rounding of the final low addition $u_\ell = t_\ell \oplus (q_\ell \oplus r)$ in step~13:}
        Since $|t_\ell| \le \frac{1}{2}\ulp{u_h}$ from Fast2Sum, the rounding error satisfies:
        \eq{}{
            \abs{u_\ell - (t_\ell + (q_\ell \oplus r))} \le 2^{-p} \abs{u_\ell} \le \frac{1}{4} \ulp{\ulp{u_h}} +
            2^{-p}\abs{q_\ell \oplus r}.
        }

        \item \emph{Truncation of omitted table limbs $j \ge 4$:}
        The omitted tail contributes $\abs{x_r \sum_{j=4}^\infty D_j}$.
    \end{enumerate}
    Summing these four error contributions yields:
    \eq{}{
        \abs{\hat{u} - u} \le \frac{1}{4}\ulp{\ulp{u_h}} + \epsilon_{\text{tail}},
    }
    with $\epsilon_{\text{tail}}$ as defined in~\eqr{epsTailDef}.
    Because $|u_h| < 1$, we have $\ulp{\ulp{u_h}} \le 2^{-2p+1}$, which provides the absolute bound:
    \eq{}{
        \abs{\hat{u} - u} \le \frac{1}{4} \cdot 2^{-2p+1} + \epsilon_{\text{tail}} = 2^{-2p-1} + \epsilon_{\text{tail}}.
    }

    \emph{Proof of (5):}
    Step~\ref{alg:step14} evaluates the scaled reduced argument $\hat{y} = y_h + y_\ell \approx \hat{u} \cdot
    \frac{\pi}{2^N}$ via double-double multiplication with a precomputed double-double constant $P = P_h + P_\ell
    \approx \frac{\pi}{2^N}$. The total error consists of three contributions:
    \begin{enumerate}
        \item \emph{Propagated error from $\hat{u}$:}
        Scaling the error bound of $\hat{u}$ by $P$ gives:
        \eq{}{
            \abs{\hat{u} - u} \cdot \abs{P} \le \abs{\hat{u} - u} \cdot \frac{\pi}{2^N}.
        }

        \item \emph{Double-double multiplication error:}
        Evaluating $(u_h + u_\ell) \cdot (P_h + P_\ell)$ via \texttt{quick\_mult} incurs an arithmetic error bounded by:
        \eq{}{
            \abs{\prt{y_h + y_\ell} - \hat{u} \cdot P} \le \epsilon_{\text{mult}}.
        }

        \item \emph{Constant approximation error:}
        The precomputed constant satisfies $|P - \frac{\pi}{2^N}| \le \epsilon_P$. Scaling by $|u|$ gives:
        \eq{}{
            \abs{P - \frac{\pi}{2^N}} \cdot \abs{u} \le \epsilon_P \abs{u} \le \epsilon_P \prt{\frac{1}{2} + \Delta_k}.
        }
    \end{enumerate}
    Summing these three contributions establishes~\eqr{yBoundGen}.
\end{proof}

\begin{Cor}[LLVM libc Parameters and Concrete Error Bounds] \lab{cor:llvm_params}
    For the LLVM libc double-precision configuration ($p = 53, p_c = 51, N = 7, M = 16, g = 9$, with $i \in [-3, 60]$
    covering $|x| \ge 2^{16}$), Algorithm~\ref{alg:main} satisfies:
    \begin{enumerate}
        \item \textbf{Quotient Separation and Residual Range:}
        The quotient separation bound is $\Delta_k < 2^{-19}$, and $\hat{k}$ is exact whenever
        $\left\{x \cdot \frac{2^N}{\pi}\right\}$ is separated from $\pm 1/2$ by more than $2^{-19}$.
        The intermediate and exact residuals satisfy:
        \eq{}{
            \abs{v} \le \frac{1}{2} + 2^{-21}, \quad \text{and} \quad \abs{u} < \frac{1}{2} + 2^{-19}.
        }

        \item \textbf{Fast2Sum Exactness:}
        Both additions $\fts{b_\ell, c_h}$ and $\fts{v, q_h}$ are exact.

        \item \textbf{Double-Double Residual Accuracy:}
        With tail error $\epsilon_{\text{tail}} < 2^{-123} + 2^{-126}$, the computed residual $\hat{u} = u_h + u_\ell$
        satisfies:
        \eq{uBound}{
        \begin{aligned}
            \abs{\prt{u_h + u_\ell} - \prt{x \cdot \frac{2^N}{\pi} - \hat{k}} \bmod 2^{N+1}}
                &< \frac{1}{4}\ulp{\ulp{u_h}} + 2^{-123} + 2^{-126} \\
                &\le \frac{1}{2} \max\bigl(\ulp{\ulp{u_h}}, 2^{-120}\bigr) \le 2^{-106}.
        \end{aligned}
        }
        In particular, because $|u_h| < 1$, we have $\ulp{\ulp{u_h}} \le 2^{-105}$, giving the sharp absolute bound:
        \eq{}{
            \abs{\hat{u} - u} < 2^{-107} + 2^{-122}.
        }

        \item \textbf{Accuracy of the Scaled Reduced Argument:}
        With the precomputed constant $P_h + P_\ell \approx \pi/128$ for $N=7$ ($|P - \pi/128| \le 2^{-115}$), the
        scaled reduced argument $\hat{y} = y_h + y_\ell$ in step~\ref{alg:step14} satisfies:
        \eq{yBound}{
            \abs{\prt{y_h + y_\ell} - \prt{x - \hat{k} \cdot \frac{\pi}{2^N}} \bmod 2\pi} < 2^{-110}.
        }
    \end{enumerate}
\end{Cor}

\begin{proof}
    We apply Theorem~\ref{thm:correctness} with $p = 53$, $p_c = 51$, $N = 7$, $M = 16$, and $g = 9$. Across the
    table domain $i \in [-3, 60]$, the scaled input satisfies $|x_r| < 2^{78}$, and evaluating all 64 table
    entries gives:
    \eq{}{
        \ufp{a_\ell \oplus b_h} \le 2^{32}, \quad \ufp{b_h} \le 2^{31}, \quad \abs{c_h} < 0.9 \cdot 2^{-20},
        \quad \abs{x_r D_3} < 2^{-72},
    }
    and the omitted tail is bounded by:
    \eq{}{
        \abs{x_r \sum_{j=4}^\infty D_j} < 2^{-126},
    }
    which is $12$ bits below the worst-case budget $\varepsilon = \ulp{w^*} = 2^{-114}$ of
    Section~\ref{sec:WorstCase}.

    \emph{Proof of (1):}
    Since $p = 53$, the first two terms of~\eqr{DeltaKDef} satisfy:
    \eq{}{
    \begin{aligned}
        \frac{1}{2}\ulp{a_\ell \oplus b_h} &= 2^{-53} \ufp{a_\ell \oplus b_h} \le 2^{-21}, \\
        \abs{b_\ell} &\le \frac{1}{2}\ulp{b_h} = 2^{-53} \ufp{b_h} \le 2^{-22}.
    \end{aligned}
    }
    Since $|c_h| < 2^{-20}$, we have $\ulp{c_h} \le 2^{-73}$ and $|c_\ell| \le \frac{1}{2}\ulp{c_h} \le 2^{-74}$.
    Treating the omitted tail as a single limb, the remaining terms of~\eqr{DeltaKDef} satisfy:
    \eq{}{
    \begin{aligned}
        \abs{x_r} \sum_{j=2}^\infty \abs{D_j}
            &\le \abs{c_h} + \abs{c_\ell} + \abs{x_r D_3} + \abs{x_r \sum_{j=4}^\infty D_j} \\
            &< 0.9 \cdot 2^{-20} + 2^{-74} + 2^{-72} + 2^{-126} \\
            &< 2^{-20}.
    \end{aligned}
    }
    Summing these bounds yields:
    \eq{}{
        \Delta_k < 2^{-21} + 2^{-22} + 2^{-20} < 2^{-19}.
    }
    By Theorem~\ref{thm:correctness}(2), $\hat{k}$ is exact whenever $\left\{x \cdot \frac{2^N}{\pi}\right\}$ is
    separated from $\pm 1/2$ by more than $2^{-19}$, and:
    \eq{}{
        \abs{v} \le \frac{1}{2} + \frac{1}{2}\ulp{a_\ell \oplus b_h} \le \frac{1}{2} + 2^{-21}, \quad
        \abs{u} \le \frac{1}{2} + \Delta_k < \frac{1}{2} + 2^{-19}.
    }

    \emph{Proof of (2):}
    By Lemma~\ref{Lem:LsbLoHi}, every input satisfies $\uls{b_\ell} \ge \ulp{c_h}$, so $\fts{b_\ell, c_h}$ is exact
    by Theorem~\ref{thm:correctness}(3).
    Since $p - p_c = 2$, the proof of Theorem~\ref{thm:correctness}(3) gives $|c_h| < 4\ulp{b_h}$, so:
    \eq{}{
        \abs{q_h} = \abs{b_\ell \oplus c_h} \le \frac{1}{2}\ulp{b_h} + 4\ulp{b_h} < 2^3 \ulp{b_h},
    }
    where rounding preserves the bound because $\frac{1}{2}\ulp{b_h} + 4\ulp{b_h} \in \F$. Therefore:
    \eq{}{
        \ulp{q_h} \le 2^{-52} \cdot 2^2 \ulp{b_h} = 2^{-50} \ulp{b_h} < \ulp{b_h} \le \uls{v},
    }
    and $\fts{v, q_h}$ is also exact.

    \emph{Proof of (3):}
    We bound the four terms of~\eqr{epsTailDef}. From the proof of (1), the input of step~11 satisfies:
    \eq{}{
        \abs{c_\ell + x_r D_3} < 2^{-74} + 2^{-72}.
    }
    Since $2^{-74} + 2^{-72} \in \F$, rounding preserves this bound, so $|r| < 2^{-71}$ and:
    \eq{}{
        \frac{1}{2}\ulp{r} = 2^{-53} \ufp{r} \le 2^{-53} \cdot 2^{-72} = 2^{-125}.
    }
    Next, $|b_\ell + c_h| < 2^{-22} + 2^{-20} < 2^{-19}$, so $\ufp{q_h} \le 2^{-20}$ and:
    \eq{}{
        \abs{q_\ell} \le \frac{1}{2}\ulp{q_h} \le 2^{-73}.
    }
    Hence $|q_\ell + r| \le 2^{-73} + 2^{-74} + 2^{-72} < 2^{-71}$, and again rounding preserves this bound, so:
    \eq{}{
        \frac{1}{2}\ulp{q_\ell \oplus r} \le 2^{-125}, \quad 2^{-p} \abs{q_\ell \oplus r} < 2^{-53} \cdot 2^{-71}
        = 2^{-124}.
    }
    Adding the omitted tail, the tail error satisfies:
    \eq{}{
        \epsilon_{\text{tail}} < 2^{-125} + 2^{-125} + 2^{-124} + 2^{-126} = 2^{-123} + 2^{-126}.
    }
    Substituting into~\eqr{uBoundGen} gives the first line of~\eqr{uBound}. For the second line, let $U :=
    \ulp{\ulp{u_h}}$, which is a power of two. If $U \ge 2^{-120}$, then:
    \eq{}{
        \frac{1}{4} U + 2^{-123} + 2^{-126} \le \prt{2^{-2} + 2^{-3} + 2^{-6}} U < \frac{1}{2} U.
    }
    If $U \le 2^{-121}$, then:
    \eq{}{
        \frac{1}{4} U + 2^{-123} + 2^{-126} \le 2^{-123} + 2^{-123} + 2^{-126} < 2^{-121} = \frac{1}{2} \cdot 2^{-120}.
    }
    Finally, since $|u_h| < 1$, we have $U \le 2^{-105}$, which gives the last inequality of~\eqr{uBound} and:
    \eq{}{
        \abs{\hat{u} - u} < 2^{-107} + 2^{-123} + 2^{-126} < 2^{-107} + 2^{-122}.
    }

    \emph{Proof of (4):}
    We bound the three terms of~\eqr{yBoundGen}. Since $\pi/128 < 2^{-5}$, the propagated error satisfies:
    \eq{}{
        \abs{\hat{u} - u} \cdot \frac{\pi}{128} < \prt{2^{-107} + 2^{-122}} \cdot 2^{-5} = 2^{-112} + 2^{-127}.
    }
    For the double-double multiplication (\texttt{quick\_mult}), $|u_\ell| \le 2^{-54}$ and $|P_\ell| < 2^{-59}$,
    so the omitted product satisfies:
    \eq{}{
        \abs{u_\ell P_\ell} < 2^{-54} \cdot 2^{-59} = 2^{-113}.
    }
    Adding the floating-point rounding errors, which are at most $2^{-113} + 2^{-112}$, gives:
    \eq{}{
        \epsilon_{\text{mult}} < 2^{-113} + 2^{-113} + 2^{-112} = 2^{-111}.
    }
    For the constant, $\epsilon_P = |P - \pi/128| \le 2^{-115}$, and by item~(1):
    \eq{}{
        \epsilon_P \prt{\frac{1}{2} + \Delta_k} < 2^{-115} \prt{\frac{1}{2} + 2^{-19}} < 2^{-115}.
    }
    Summing these three contributions yields:
    \eq{}{
    \begin{aligned}
        \abs{\hat{y} - y} &< 2^{-112} + 2^{-127} + 2^{-111} + 2^{-115} \\
            &< 2^{-112} + 2^{-111} + 2^{-114} < 2^{-110},
    \end{aligned}
    }
    which establishes~\eqr{yBound}.
\end{proof}

\begin{Rem}[Accuracy for single-stage and fast-path use] \lab{rem:fastpath}
    The four limbs of Algorithm~\ref{alg:main} meet the worst-case budget $\varepsilon = \ulp{w^*}$ of
    Section~\ref{sec:WorstCase}, which is needed when the result is used without a rounding test; in LLVM
    libc, the accurate path also reuses $D_3$. Since $|u| \ge w^* > 2^{-62}$ for every input with $|k| \ge 1$,
    and $\frac{1}{4}\ulp{\ulp{u_h}} \le 2^{-106} |u_h| \le 2^{-105} |u|$, \eqr{uBound} also gives a relative
    error on $u$ below:
    \eq{}{
        2^{-105} + \frac{2^{-123} + 2^{-126}}{2^{-62}} = 2^{-105} + 2^{-61} + 2^{-64} < 2^{-60}
    }
    for every input that requires reduction.
    The fourth limb costs 8 bytes per table entry (2~KB instead of 1.5~KB), and it requires the product $x_r D_2$ to
    be computed exactly. The fast path of a correctly rounded implementation only needs the budget $\varepsilon =
    2^{-67}$, which the same table meets with three limbs by~\eqr{MBoundEps}. Then $x_r D_2$ can be merged with
    $b_\ell$ in a single FMA, which replaces the 11 floating-point operations of steps~7 and~10--13 by 4.
    Appendix~\ref{sec:three_limb} gives the resulting three-limb variant (Algorithm~\ref{alg:three},
    Figure~\ref{fig:pipeline3}), and Corollary~\ref{cor:three_limb} shows that for $N = 7$:
    \eq{}{
        \abs{\hat{u} - u} < 2^{-72} + 2^{-73}, \quad \abs{\hat{y} - y} < 2^{-77} + 2^{-78} + 2^{-110}.
    }
    This is tighter than the fast-path reductions of CORE-MATH ($< 2^{-73.3}$ for \texttt{reduce\_fast} and
    $< 2^{-66.3}$ for \texttt{reduce\_large}, Table~\ref{tab:AllAlgo}), which likewise rely on Ziv's rounding test
    for rare cancellation cases. Conversely, keeping four limbs, the fast-path budget allows $M = 32$, which halves
    the table but loosens the quotient estimate $\hat{k}$. We leave a systematic study of these trade-offs,
    including the resulting fraction of inputs that reach the accurate path, to future work.
    \RemEnd
\end{Rem}

\section{Directed Rounding Modes and FMA-Free Targets} \lab{sec:rounding}

While Section~\ref{sec:Main} assumes round-to-nearest ($\operatorname{RN}$), in some applications other IEEE
754 rounding modes ($\operatorname{RU}, \operatorname{RD}, \operatorname{RZ}$) are used, and some target
architectures lack hardware FMA instructions. We discuss how the algorithm behaves in these scenarios.

\subsection{Directed Rounding Modes} \lab{subsec:dir_round}

When hardware FMA is available, Algorithm~\ref{alg:main} remains accurate under directed rounding modes:
\begin{enumerate}
    \item \emph{Exact operations:}
    By Lemma~\ref{Lem:UfpUlsProp} and Theorem~\ref{thm:FirstSteps}, the difference $a_\ell - \hat{k}$ and
    addition $(a_\ell - \hat{k}) + b_h$ are both exactly representable in $\F$ across all rounding modes.
    Because the exact values already lie in $\F$, rounding has no effect: $\circ(a_\ell - \hat{k}) = a_\ell -
    \hat{k}$ and $\circ((a_\ell - \hat{k}) + b_h) = (a_\ell - \hat{k}) + b_h$ for all rounding modes $\circ \in
    \{\operatorname{RN}, \operatorname{RU}, \operatorname{RD}, \operatorname{RZ}\}$.
    Theorem~\ref{thm:FirstSteps} assumes that step~3 is rounded to nearest. In a directed mode, the rounding
    error of step~3 satisfies:
    \eq{}{
        \abs{(a_\ell \oplus b_h) - (a_\ell + b_h)} < 2\max\prt{\ulp{a_\ell}, \ulp{b_h}} \le \frac{1}{2},
    }
    which is at most $2^{-20}$ for the LLVM table. With the static-$\operatorname{RN}$ $\hat{k}$ of item~4, this
    gives $|(a_\ell + b_h) - \hat{k}| < 1$, and the proof of Theorem~\ref{thm:FirstSteps} applies with $3/4$
    replaced by $1$. With a dynamic-mode \texttt{rint}, the bound becomes $1 + 2^{-20}$, and the proof still
    applies because every table entry satisfies:
    \eq{}{
        \uls{a_\ell} \ge \ulp{b_h} \ge 2^{-46} \ge 2^{-p+2}.
    }
    
    \item \emph{Exact product low part:}
    Boldo and Daumas~\cite{BD03} showed that for any $a, b \in \F$, the mathematical error $a \cdot b -
    \circ(a \cdot b)$ is representable in $p$ bits under all IEEE 754 rounding modes. Therefore, $\fma(a, b,
    -\circ(a \cdot b))$ computes the exact low part of the product under $\operatorname{RU},
    \operatorname{RD}$, and $\operatorname{RZ}$.
    
    \item \emph{Fast2Sum and roundings in directed modes:}
    Under directed roundings, each rounding error in the proof of Theorem~\ref{thm:correctness}(4) is bounded by
    one ulp instead of half an ulp. The two Fast2Sum operations, however, remain exact. Park, Lim, and
    Nagarakatte~\cite[Theorem~4]{PLN26} showed that $\fts{a, b}$ is exact under all faithful rounding modes,
    including $\operatorname{RU}$, $\operatorname{RD}$, and $\operatorname{RZ}$, provided that no overflow occurs
    and:
    \eq{PLNcond}{
        a \in \ulp{b} \cdot \Z, \quad b \in 2^{-2p+1} \ufp{a} \cdot \Z.
    }
    The first condition is $\uls{a} \ge \ulp{b}$, as in the round-to-nearest criterion of~\cite{JZ25}. For the
    LLVM table, both Fast2Sum operations satisfy~\eqr{PLNcond}:
    \begin{itemize}
        \item \emph{Fast2Sum in step~10:} The proof of Lemma~\ref{Lem:LsbLoHi} only uses the monotonicity of
            rounding, so $\uls{b_\ell} \ge \ulp{c_h}$ holds in every rounding mode. For the second condition,
            $c_h$ is an integer multiple of $\ulp{c_h}$. Since $|x_r D_1| < 4 \ufp{x_r} \ufp{D_1}$ and $|x_r D_2|
            \ge \ufp{x_r} \ufp{D_2}$, we have:
            \eq{}{
                \ufp{b_\ell} < \ulp{b_h} \le 2^{-50} \ufp{x_r} \ufp{D_1}, \quad
                \ulp{c_h} \ge 2^{-52} \ufp{x_r} \ufp{D_2}.
            }
            Evaluating all 64 table entries gives $\ufp{D_1} \le 2^{60} \ufp{D_2}$, so:
            \eq{}{
                2^{-105} \ufp{b_\ell} < 2^{-155} \ufp{x_r} \ufp{D_1} \le 2^{-95} \ufp{x_r} \ufp{D_2} < \ulp{c_h}.
            }
            If $b_\ell = 0$, the Fast2Sum is trivially exact.
        \item \emph{Fast2Sum in step~12:} $v$ is an integer multiple of $\ulp{b_h}$. Since $|b_\ell| <
            \ulp{b_h}$ and $|c_h| \le 4 \ulp{b_h}$ (proof of Corollary~\ref{cor:llvm_params}~(2)), we have
            $|q_h| \le 8 \ulp{b_h}$, so $\ulp{q_h} \le 2^{-49} \ulp{b_h}$ and $\uls{v} \ge \ulp{q_h}$ when $v \ne
            0$. For the second condition, $b_\ell$ and $c_h$ are integer multiples of $\ulp{c_h}$, hence so is
            $q_h$. Evaluating all 64 table entries with $|x_r| \ge 2^{62}$ gives $\ulp{c_h} \ge 2^{-100}$. Since
            $|v| < 1 + 2^{-20}$ by item~1, $\ufp{v} \le 1$ and:
            \eq{}{
                2^{-105} \ufp{v} \le 2^{-105} < \ulp{c_h}.
            }
            If $v = 0$, the Fast2Sum is trivially exact.
    \end{itemize}
    For the LLVM table, the remaining error terms in the proof of Corollary~\ref{cor:llvm_params} become:
    \begin{itemize}
        \item Since $|c_\ell| < \ulp{c_h} \le 2^{-73}$, we have $|c_\ell + x_r D_3| < 2^{-73} + 2^{-72} < 2^{-71}$,
            so the rounding error of $r$ in step~11 is below $2^{-124}$.
        \item Since $|q_\ell| < \ulp{q_h} \le 2^{-72}$, we have $|q_\ell + r| < 2^{-70}$, so the rounding error
            of $q_\ell \oplus r$ is below $2^{-123}$.
        \item Since $|t_\ell| < \ulp{u_h}$ and $|q_\ell \oplus r| \le 2^{-70}$, the rounding error of $u_\ell = t_\ell
            \oplus (q_\ell \oplus r)$ is below $\ulp{\ulp{u_h}} + 2^{-122}$.
        \item The omitted tail is below $2^{-126}$ as before.
    \end{itemize}
    Summing these terms with $U := \ulp{\ulp{u_h}}$ yields:
    \eq{}{
    \begin{aligned}
        \abs{\hat{u} - u} &< U + 2^{-122} + 2^{-123} + 2^{-124} + 2^{-126} \\
            &< U + 2^{-121} \le \frac{3}{2} \max\bigl(U, 2^{-120}\bigr).
    \end{aligned}
    }
    When $\hat{k}$ uses static $\operatorname{RN}$ rounding, $|u_h| < 1$ and $U \le 2^{-105}$, so:
    \eq{}{
        \abs{\hat{u} - u} < 2^{-105} + 2^{-121}.
    }
    With dynamic \texttt{rint} rounding, $|u_h| < 2$ and $U \le 2^{-104}$, so $|\hat{u} - u| < 2^{-104} + 2^{-121}$.

    \item \emph{Rounding of quotient $\hat{k}$:}
    In LLVM libc~\cite{L}, on architectures supporting static rounding instructions (such as SSE4.1
    \texttt{roundsd} on x86-64 and \texttt{frintn} on AArch64), $\hat{k} = \nearestint{a_\ell \oplus b_h}$ is
    computed using round-to-nearest regardless of the dynamic rounding mode. Because $a_\ell + b_h$ and $a_\ell
    \oplus b_h$ are both integer multiples of $\ulp{b_h}$, the directed rounding error is bounded by:
    \eq{}{
        \abs{(a_\ell \oplus b_h) - (a_\ell + b_h)} \le \ulp{a_\ell \oplus b_h} - \ulp{b_h}.
    }
    Adding $|b_\ell| < \ulp{b_h}$, $|c_h| < 0.9 \cdot 2^{-20}$, and the remaining tail terms, which are below
    $2^{-71}$, yields:
    \eq{}{
    \begin{aligned}
        \abs{u} &< \frac{1}{2} + \ulp{a_\ell \oplus b_h} + 0.9 \cdot 2^{-20} + 2^{-71} \\
            &\le \frac{1}{2} + 2^{-20} + 0.9 \cdot 2^{-20} + 2^{-71} < \frac{1}{2} + 2^{-19}.
    \end{aligned}
    }
    On platforms using a generic \texttt{rint} fallback under directed rounding, $\hat{k}$ rounds in the current
    dynamic mode, yielding $|u| < 1 + 2^{-19}$, which is readily accommodated by the subsequent evaluation
    polynomial.
\end{enumerate}

\subsection{FMA-Free Platforms and Veltkamp/Dekker Splitting} \lab{subsec:nofma}

When hardware FMA is not available (such as on older x86 processors, ARMv7 without VFPv4, or baseline
WebAssembly), exact products $x_r \cdot D_j = p_h + p_\ell$ can be computed using Veltkamp splitting and
Dekker's multiplication algorithm~\cite{Dekker, Mu2}.

Classical Veltkamp splitting decomposes $a \in \F$ into $a = a_{\text{hi}} + a_{\text{lo}}$ using:
\eq{Veltkamp}{
    \gamma = \circ(C \cdot a), \quad \delta = a \ominus \gamma, \quad
    a_{\text{hi}} = \gamma \oplus \delta, \quad a_{\text{lo}} = a \ominus a_{\text{hi}},
}
with splitting constant $C = 2^s + 1$, where $s = \ceil{p/2}$ ($s = 27$ for double precision). While some
variants of Veltkamp splitting fail under directed roundings, \cite[Theorem~1]{Zim24} showed that the
formulation in~\eqr{Veltkamp} remains exact under all IEEE 754 rounding modes ($\operatorname{RN},
\operatorname{RU}, \operatorname{RD}, \operatorname{RZ}$), with $a_{\text{hi}}$ fitting in $p-s$ bits and
$a_{\text{lo}}$ in $s$ bits, provided $2 \le s \le p-2$.

For Dekker's product, when $p$ is even ($s = p/2$), \cite[Theorem~2]{Zim24} shows that the product $x \cdot y
= r_1 + r_2$ is exact under all IEEE 754 rounding modes. When $p$ is odd (such as $p=53$), using $s = (p+1)/2$
means the low-part product $x_{\text{lo}} \cdot y_{\text{lo}}$ has $s + s = p + 1$ bits ($27 + 27 = 54 > 53$), which
can exceed $p$ bits and introduce a 1-ulp rounding error under directed roundings~\cite[Lemma~1]{Zim24}.

However, when one operand has precision at most $p - 1$, Dekker's product is exact across all rounding modes:

\begin{Prop} \lab{prop:dekker_exact}
    Let $p \ge 9$ be an odd integer, and let $s = (p + 1)/2$. Let $x \in \F$ have precision $p$, and let $y
    \in \F$ have precision at most $p - 1$. Assuming no overflow or underflow occurs, Dekker's product using
    Veltkamp splitting with constant $C = 2^s + 1$, evaluated in precision $p$, satisfies:
    \eq{}{
        x \cdot y = r_1 + r_2
    }
    exactly across all IEEE 754 rounding modes ($\operatorname{RN}, \operatorname{RU}, \operatorname{RD},
    \operatorname{RZ}$).
\end{Prop}

\begin{proof}
    We mostly follow the steps in the proof of the corresponding theorem in~\cite{Zim24}.
    By scale and sign invariance (noting that negation swaps $\operatorname{RU}$ and $\operatorname{RD}$), we
    may assume $1 \le x, y < 2$. If $x=0$ or $y=0$, the result is trivial.

    Since $2 \le s = (p+1)/2 \le p-2$ for all odd $p \ge 5$, \cite[Theorem~1]{Zim24} applies: Veltkamp
    splitting with $C = 2^s + 1$ decomposes $x = x_{\text{hi}} + x_{\text{lo}}$ and $y = y_{\text{hi}} +
    y_{\text{lo}}$ exactly for all rounding modes, with:
    \eq{}{
        \uls{x_{\text{hi}}}, \uls{y_{\text{hi}}} \ge 2^{-(p-3)/2}, \quad
        \uls{x_{\text{lo}}} \ge 2^{-p+1}, \quad \abs{x_{\text{lo}}} \le 2^{-(p-3)/2} - 2^{-p+1}.
    }
    Both $x_{\text{hi}}$ and $y_{\text{hi}}$ fit in $p - s = (p-1)/2$ bits, and $x_{\text{lo}}$ fits in $s =
    (p+1)/2$ bits.

    For $y_{\text{lo}} = y - y_{\text{hi}}$, because $y$ has precision at most $p - 1$, $y$ is an integer
    multiple of $2^{-p+2}$. Since $y_{\text{hi}}$ is a multiple of $2^{-(p-3)/2}$, which is also a multiple of
    $2^{-p+2}$ for all $p \ge 1$, $y_{\text{lo}}$ is an integer multiple of $2^{-p+2}$.
    By~\cite[Theorem~1]{Zim24}, $|y_{\text{lo}}| < 2^{-(p-3)/2}$, so
    \eq{}{
        \abs{y_{\text{lo}}} \le 2^{-(p-3)/2} - 2^{-p+2}.
    }
    Thus the non-zero bits of $y_{\text{lo}}$ range from at most $2^{-(p-1)/2}$ down to $2^{-p+2}$, fitting in at most:
    \eq{}{
        \prt{-\frac{p-1}{2}} - (-p + 2) + 1 = \frac{p-1}{2} \text{ bits}.
    }

    The four partial products therefore satisfy:
    \begin{align*}
        \prc(x_{\text{hi}} \cdot y_{\text{hi}}) &\le \frac{p-1}{2} + \frac{p-1}{2} = p - 1 \le p, \\
        \prc(x_{\text{hi}} \cdot y_{\text{lo}}) &\le \frac{p-1}{2} + \frac{p-1}{2} = p - 1 \le p, \\
        \prc(x_{\text{lo}} \cdot y_{\text{hi}}) &\le \frac{p+1}{2} + \frac{p-1}{2} = p, \\
        \prc(x_{\text{lo}} \cdot y_{\text{lo}}) &\le \frac{p+1}{2} + \frac{p-1}{2} = p.
    \end{align*}
    Because $y_{\text{lo}}$ fits in $(p-1)/2$ bits, $x_{\text{lo}} \cdot y_{\text{lo}}$ requires at most $p$ bits
    (unlike the general case where $x_{\text{lo}} \cdot y_{\text{lo}}$ can require $p+1$
    bits~\cite[Lemma~1]{Zim24}). Hence all four products are computed exactly in precision $p$ without
    rounding:
    \begin{align*}
        \circ(x_{\text{hi}} \cdot y_{\text{hi}}) &= x_{\text{hi}} \cdot y_{\text{hi}}, \\
        \circ(x_{\text{hi}} \cdot y_{\text{lo}}) &= x_{\text{hi}} \cdot y_{\text{lo}}, \\
        \circ(x_{\text{lo}} \cdot y_{\text{hi}}) &= x_{\text{lo}} \cdot y_{\text{hi}}, \\
        \circ(x_{\text{lo}} \cdot y_{\text{lo}}) &= x_{\text{lo}} \cdot y_{\text{lo}}.
    \end{align*}

    Following the accumulation steps in~\cite[Theorem~2]{Zim24}:
    \begin{enumerate}
        \item \emph{Computation of $t_1$:}
        Here $t_1 = \circ(-r_1 + x_{\text{hi}} \cdot y_{\text{hi}})$. We have:
        \begin{align*}
            \abs{x \cdot y - x_{\text{hi}} \cdot y_{\text{hi}}}
                &\le \abs{x_{\text{lo}}} \cdot y + \abs{y_{\text{lo}}} \cdot x_{\text{hi}} < 2 \cdot 2^{-(p-3)/2} + 2
                    \cdot 2^{-(p-3)/2} = 2^{-(p-7)/2}, \\
            \abs{x \cdot y - r_1} &\le 2^{-p+2}.
        \end{align*}
        Since $r_1 = \circ(x \cdot y) \in [1, 4)$ and $x_{\text{hi}} \cdot y_{\text{hi}} \ge 1$, applying the triangle
        inequality yields:
        \eq{}{
            \abs{r_1 - x_{\text{hi}} \cdot y_{\text{hi}}} < 2^{-p+2} + 2^{-(p-7)/2}.
        }
        For odd $p \ge 9$, $2^{-p+2} + 2^{-(p-7)/2} \le 2^{-7} + 2^{-1} = 0.5078125 \le 1 \le x_{\text{hi}} \cdot
        y_{\text{hi}}$. Since $x_{\text{hi}} \cdot y_{\text{hi}} \ge 1$ and $r_1 \ge 1$, this inequality implies $r_1$
        and $x_{\text{hi}} \cdot y_{\text{hi}}$ have the same sign and satisfy:
        \eq{}{
            \frac{1}{2} \le \frac{r_1}{x_{\text{hi}} \cdot y_{\text{hi}}} \le 2.
        }
        By Sterbenz's Lemma, the subtraction is exact:
        \eq{}{
            t_1 = -r_1 + x_{\text{hi}} \cdot y_{\text{hi}}.
        }

        \item \emph{Computation of $t_2$:}
        Here $t_1 + x_{\text{hi}} \cdot y_{\text{lo}} = -r_1 + x_{\text{hi}} \cdot y = (x \cdot y - r_1) - x_{\text{lo}}
        \cdot y$. Using $|x_{\text{lo}}| \le 2^{-(p-3)/2} - 2^{-p+1}$ and $y < 2$:
        \eq{}{
            \abs{t_1 + x_{\text{hi}} \cdot y_{\text{lo}}} \le \abs{x \cdot y - r_1} + \abs{x_{\text{lo}}} \cdot y <
            2^{-p+2} + 2\prt{2^{-(p-3)/2} - 2^{-p+1}} = 2^{-(p-5)/2}.
        }
        Since $x_{\text{hi}}$ is a multiple of $2^{-(p-3)/2}$ and $y$ is a multiple of $2^{-p+2}$, the term
        $x_{\text{hi}} \cdot y$ is an integer multiple of $2^{-(3p-7)/2}$. The leading product $r_1$ is a multiple
        of $2^{-p+1}$, which is also a multiple of $2^{-(3p-7)/2}$ for $p \ge 5$. Thus $t_1 + x_{\text{hi}}
        \cdot y_{\text{lo}}$ is a multiple of $2^{-(3p-7)/2}$ bounded by $2^{-(p-5)/2}$, spanning at most $p-1$
        bits. Hence:
        \eq{}{
            t_2 = \circ(t_1 + x_{\text{hi}} \cdot y_{\text{lo}}) = -r_1 + x_{\text{hi}} \cdot y.
        }

        \item \emph{Computation of $t_3$:}
        Here $t_2 + x_{\text{lo}} \cdot y_{\text{hi}} = -r_1 + x \cdot y - x_{\text{lo}} \cdot y_{\text{lo}}$.
        Using $|y_{\text{lo}}| < 2^{-(p-3)/2}$ and $|x_{\text{lo}}| < 2^{-(p-3)/2}$:
        \eq{}{
            \abs{x_{\text{lo}} \cdot y_{\text{lo}}} < 2^{-(p-3)/2} \cdot 2^{-(p-3)/2} = 2^{-p+3}.
        }
        Therefore:
        \eq{}{
            \abs{t_2 + x_{\text{lo}} \cdot y_{\text{hi}}} \le \abs{-r_1 + x \cdot y} + \abs{x_{\text{lo}} \cdot
            y_{\text{lo}}} < 2^{-p+2} + 2^{-p+3} < 2^{-p+4}.
        }
        Also, $t_2 + x_{\text{lo}} \cdot y_{\text{hi}} = -r_1 + x \cdot y_{\text{hi}} + x_{\text{hi}}
        \cdot y_{\text{lo}}$ is an integer multiple of $2^{-(3p-5)/2}$. Bounded by $2^{-p+4}$, its non-zero
        bits span at most $(-p+3) - (-(3p-5)/2) + 1 = (p+3)/2 < p$ bits for all odd $p \ge 5$. Hence:
        \eq{}{
            t_3 = \circ(t_2 + x_{\text{lo}} \cdot y_{\text{hi}}) = -r_1 + x \cdot y - x_{\text{lo}} \cdot y_{\text{lo}}.
        }

        \item \emph{Computation of $r_2$:}
        Since $\circ(x_{\text{lo}} \cdot y_{\text{lo}}) = x_{\text{lo}} \cdot y_{\text{lo}}$, we have
        $t_3 + \circ(x_{\text{lo}} \cdot y_{\text{lo}}) = -r_1 + x \cdot y$, with:
        \eq{}{
            \abs{-r_1 + x \cdot y} \le 2^{-p+2}.
        }
        Because $x$ is a multiple of $2^{-p+1}$ and $y$ is a multiple of $2^{-p+2}$, both $x \cdot y$ and $r_1$ are
        integer multiples of $2^{-2p+3}$. A multiple of $2^{-2p+3}$ bounded by $2^{-p+2}$ spans at most $p$
        bits, so:
        \eq{}{
            r_2 = \circ(t_3 + \circ(x_{\text{lo}} \cdot y_{\text{lo}})) = -r_1 + x \cdot y,
        }
        and
        \eq{}{
            r_1 + r_2 = x \cdot y.
        }
    \end{enumerate}
\end{proof}

For double precision ($p=53$), Proposition~\ref{prop:dekker_exact} guarantees that Dekker's product is exact
across all IEEE 754 rounding modes for any second operand with precision $p_c \le 52$ bits when using the
standard splitting constant $C = 2^{27} + 1$. In LLVM libc's FMA-free implementation~\cite{L}, this property
is used as follows:
\begin{itemize}
    \item The table constants $D_j$ are generated with 51-bit precision ($p_c = 51 \le 52$) for $D_0, D_1, D_2$,
        and 53-bit for $D_3$, shared between the FMA and non-FMA paths.
    \item In the non-FMA path, both $x_r$ and the table constants are split with the standard constant
        $C = 2^{27}+1$ ($s = 27$). Since $D_0, D_1, D_2$ have $p_c = 51 \le 52$ bits (and $P_h$ has 50),
        Proposition~\ref{prop:dekker_exact} makes every Dekker product exact under all rounding modes.
    \item The tail product $x_r D_3$ need not be exact: without hardware FMA, step~11 uses two roundings via
        an emulated multiply-add, which adds at most $2^{-126}$ to $\epsilon_{\text{tail}}$. The tail error is then
        below $2^{-123} + 2^{-125}$, so the bound $\frac{1}{2}\max\bigl(\ulp{\ulp{u_h}}, 2^{-120}\bigr)$
        in~\eqr{uBound} still holds.
    \item The input $x_r$ is split once into $x_{\text{hi}} + x_{\text{lo}}$, and each limb product is
        computed using Dekker's formula.
\end{itemize}
With 51-bit limbs, four limbs are sufficient to satisfy the worst-case error budget ($L_{\min} = 4$ in
Table~\ref{tab:limbs}, or three limbs for the fast-path budget), allowing the algorithm to work accurately on
targets without FMA.

\section{Performance and Architectural Evaluation} \lab{sec:perf}

Table~\ref{tab:AllAlgo} summarizes the architectural characteristics of major double-precision range reduction
implementations.

\begin{table}[htbp]
    \caption{Large Range Reduction Characteristics} \lab{tab:AllAlgo}
    \begin{center}
    \setlength{\tabcolsep}{2.2pt}
    \begin{tabular}{lccccc}
    \toprule
        Algorithm & Domain & Table & Type & Branching & Error Bound \\
    \midrule
        Brisebarre~\cite{BDKMR} & $|x| < 2^{63}$ & $\sim 24$ KB & Double & Few (2--3) & $\le 2^{-100} + 2^{-148}$ \\
        FreeBSD~\cite{F} & $|x| \ge 2^{19}\pi$ & $\sim 0.3$ KB & Int+FP & High ($>5$) & $< 2^{-72}$ \\
        glibc ($\le$2.43)~\cite{I} & $|x| \ge 2^{25}\pi$ & $\sim 0.6$ KB & Double & Mod. (5) & $< 2^{-91.5}$ \\
        CRLIBM~\cite{CR} & $|x| \ge 2^{39.7}$ & $\sim 0.2$ KB & Int+FP & High ($\sim 9$) & $< 2^{-112}$ \\
        CORE-MATH (\texttt{fast})~\cite{CM} & $|x| > 2\pi$ & $\sim 0.2$ KB & Int128 & Mod. (5) & $< 2^{-73.3}$ \\
        CORE-MATH (\texttt{large})~\cite{Zim26} & $|x| \ge 2^{31}$ & $\sim 0.2$ KB & Int128+FP & Few (1) &
            $< 2^{-66.3}$ \\
        SLEEF~\cite{SLEEF} & $|x| \ge 10^{14}$ & $\sim 30$ KB & Double & None (0) & $< 2^{-102.5}$ \\
        LLVM libc (Ours) & $|x| \ge 2^{16}$ & $\sim 2.0$ KB & Double & None (0) & $< 2^{-110}$ \\
    \bottomrule
    \end{tabular}
    \end{center}
\end{table}

\begin{Rem} \lab{rem:algo_characteristics}
    The Branching column in Table~\ref{tab:AllAlgo} counts the data-dependent conditional branches executed per
    call inside the Payne--Hanek routine itself, excluding the small-argument dispatch and out-of-line callees
    (e.g., \texttt{scs\_mul}). We derived these counts from the source code and checked them against the x86-64
    assembly produced by the compilers used in Table~\ref{tab:perf}; conditionals compiled to selects (as in SLEEF)
    are not counted, and the count for Brisebarre et al.\ follows their algorithm description since no code is
    available. As a complementary, compiler-independent measure, the (lizard-style) McCabe cyclomatic complexity
    of these routines is 1 for ours, 2 for CORE-MATH's \texttt{reduce\_large}, $3 + 7$ for \texttt{reduce\_fast}
    and its \texttt{set\_dd} helper, $4 + 3$ for SLEEF's \texttt{rempi} and \texttt{rempisub} (mostly ternary
    selects), 12 for glibc's \texttt{\_\_branred}, 28 for CRLIBM's \texttt{rem\_pio256\_scs}, and 46 for FreeBSD's
    \texttt{\_\_kernel\_rem\_pio2}. These static counts also include loops, many with fixed trip counts that are
    perfectly predictable after unrolling, so the table reports executed branches instead.

    The Error Bound column reports an upper bound on the absolute error of the
    returned reduced argument in radians (a double-double for all rows except CORE-MATH \texttt{reduce\_large},
    which returns a single \texttt{double} for $N = 14$).
    Bounds for FreeBSD, glibc, CRLIBM, and SLEEF are our derivations from their source code; the Brisebarre et al.\
    and CORE-MATH bounds follow their publications and source comments, while ours is established in
    Theorem~\ref{thm:correctness} and Corollary~\ref{cor:llvm_params}. In detail:
    \begin{itemize}
        \item \textbf{Brisebarre et al.~\cite{BDKMR}:} Their table-based method covers $8 \le |x| < 2^{63}$,
            storing each $(2^{8i}w) \bmod \pi/2$ as three doubles (153 bits, $\sim 24$~KB). The intermediate
            3-double sum has absolute error $< 2^{-148}$. Compressing to double-double rounds $R_{med} + R_{lo}$
            to one double in the main branch ($|R_{hi}| > 2^{-p_B}$, with parameter $p_B = 14$), giving an
            absolute error of at most $2^{-100} + 2^{-148}$ by their analysis, while an alternative summation
            order for $|R_{hi}| \le 2^{-p_B}$ bounds relative error by $\le 2^{-86}$ on $|x| < 2^{63}$. Arguments
            with $|x| \ge 2^{63}$ fall back to Payne--Hanek.
        \item \textbf{FreeBSD~\cite{F}:} The large-argument kernel \texttt{\_\_kernel\_rem\_pio2} is invoked for
            $|x| \ge 2^{19}\pi \approx 2^{20}\pi/2$. With \texttt{prec}$\,=1$ (53-bit target), it computes five
            24-bit chunks of $x \cdot 2/\pi$; the returned double-double $y[0] + y[1]$ has absolute error
            $< 2^{-72}$. It dynamically detects leading-fraction cancellation and fetches additional chunks of
            $2/\pi$, guaranteeing relative error $< 2^{-70.7}$ across the entire domain.
        \item \textbf{glibc ($\le 2.43$)~\cite{I}:} The routine \texttt{\_\_branred} is called for $|x| \ge
            105414350 \approx 2^{25}\pi$, multiplying two 26-bit halves of $x$ by 144 bits of $2/\pi$. Its
            absolute error is $< 2^{-91.5}$, but because it omits dynamic cancellation checks, its
            relative error reaches $2^{-35.4}$ near multiples of $\pi/2$ (causing $>10^5$ ulp errors in
            \texttt{sin}/\texttt{cos}, fixed in glibc~2.44 (Bug~34376, commit~\texttt{5c47592a}) by
            replacing \texttt{\_\_branred} with fdlibm's \texttt{\_\_kernel\_rem\_pio2}).
        \item \textbf{CRLIBM~\cite{CR}:} In \texttt{trigo\_fast.c}, inputs with $|x| \ge 8.64 \times 10^{11}
            \approx 2^{39.7}$ (\texttt{XMAX\_}\allowbreak\texttt{DDRR}), as well as inputs with $|x| \ge 1.32 \times
            10^7$ and $k \equiv 0 \pmod{128}$, invoke \texttt{Range\-Reduction\-SCS}
            (\texttt{rem\_pio256\_}\allowbreak\texttt{scs}),
            which reduces $x$ modulo $\pi/256$ ($N = 8$) using Software Carry-Save (SCS) arithmetic (eight 30-bit
            words in radix $2^{30}$). It splits $x$ into up to three 30-bit words, multiplies by eleven 30-bit words
            of $256/\pi$ from a 48-entry table ($\approx 0.24$~KB including $\pi/256$), propagates carries, shifts out
            up to two cancelled 30-bit words (covering up to 62 bits of cancellation), multiplies by $\pi/256$ in SCS
            (\texttt{scs\_mul}, relative error $< 2^{-200}$), and converts the top four 30-bit words into a
            double-double $y_h + y_\ell$ with relative error $\le 2^{-105}$ and absolute error $< 2^{-105} \cdot
            \pi/512 < 2^{-112}$ (conservatively bounded by $2^{-100}$ in~\cite{CR}).
        \item \textbf{CORE-MATH~\cite{CM, Zim26}:} We consider two fast-path Payne--Hanek implementations from
            CORE-MATH, both indexing a 20-entry table of 64-bit limbs of $1/(2\pi)$ ($160$~B, $M = 64$) and
            returning the residual as a fraction of a full turn ($x/(2\pi) \bmod 1$):
            \begin{enumerate}
                \item In \texttt{reduce\_fast} (commit~\texttt{9239a6b0}, used in \texttt{sin}, \texttt{cos},
                    \texttt{sincos}, and \texttt{tan}), inputs with $|x| > 2\pi$ multiply the 53-bit significand
                    $m$ by three consecutive table words (two when $|x| < 2^{52}$) in 192-bit integer arithmetic,
                    shift the fractional part into two 64-bit words, and normalize it into a double-double via
                    \texttt{set\_dd} (which uses count-leading-zeros and branches) for $N = 10$, with turn error
                    $< 2^{-75.998}$ (radian error $< 2\pi \cdot 2^{-75.998} < 2^{-73.3}$).
                \item In the newer \texttt{reduce\_large} (commit~\texttt{29f5112d}~\cite{Zim26}, used in
                    \texttt{sin} for $|x| \ge 2^{31}$), three table words are shifted into two 64-bit words (its only
                    data-dependent branch skips this shift when the bit offset is zero, avoiding a shift by 64) and
                    multiplied by $m$ in a single $53 \times 128 \to 128$-bit unsigned product. Adding a rounding
                    constant extracts both the 15-bit quotient $k$ ($N = 14$) and a single \texttt{double} residual
                    $r$ ($|r| \le 2^{-16}$) via one integer-to-\texttt{double} conversion and one FMA, without
                    \texttt{set\_dd}. Its turn error is $< 2^{-68.988}$~\cite[Lemma~4]{Zim26}, giving a radian error
                    $< 2\pi \cdot 2^{-68.988} < 2^{-66.3}$. To avoid a $2^{14}$-entry table, reconstruction uses an
                    8~KB bipartite table (at steps $\pi/2^7$ and $\pi/2^{14}$).
            \end{enumerate}
            In both routines, rare cancellation cases are caught by Ziv's rounding test and resolved by an accurate
            path.
        \item \textbf{SLEEF~\cite{SLEEF}:} In \texttt{sin}, \texttt{cos}, \texttt{sincos}, and \texttt{tan\_u10},
            \texttt{rempi} is called for $|x| \ge 10^{14}$ (and for $|x| \ge 10^6$ in \texttt{tan\_u35}),
            multiplying by a four-double expansion from a 969-entry ($\approx 30$~KB) table. The dominant error
            arises from tail rounding in the second double-double addition, giving an absolute error bound
            $< 2^{-102.5}$ and a worst-case relative error of $2^{-52.25}$ on cancellation inputs.
        \item \textbf{LLVM libc (Ours):} Algorithm~\ref{alg:main} operates for $|x| \ge 2^{16}$. It achieves an
            absolute error bound $\abs{\hat{y} - \prt{x - \hat{k} \cdot \frac{\pi}{2^N}} \bmod 2\pi} < 2^{-110}$
            (normalized residual error bounded by $2^{-107} + 2^{-122} < 2^{-105}$) without branches in
            a single pass with a 2.0~KB table. Note that the bounds in Table~\ref{tab:AllAlgo} are not directly
            comparable across all rows: FreeBSD and CRLIBM guarantee relative error bounds through cancellation
            shifting, while the CORE-MATH bounds are for fast paths backed by Ziv's rounding test and an accurate
            path. Since our truncation error meets the worst-case budget, the relative error on $u$ stays below
            $2^{-60}$ even for the inputs closest to multiples of $\pi/128$ (Remark~\ref{rem:fastpath}).
    \end{itemize}
    \RemEnd
\end{Rem}

\textbf{Performance Benchmarks on Modern x86-64:}
We benchmarked the range reduction routines of FreeBSD (\texttt{\_\_ieee754\_}\allowbreak\texttt{rem\_pio2}), glibc
(\texttt{\_\_bran}\allowbreak\texttt{red}, used through glibc~2.43), CRLIBM (\texttt{RangeReductionSCS} in
\texttt{trigo\_fast.c}~\cite{CR}), CORE-MATH (\texttt{reduce\_fast} at commit~\texttt{9239a6b0} and
\texttt{reduce\_large} at commit~\texttt{29f5112d}~\cite{CM, Zim26}), and LLVM libc (Algorithm~\ref{alg:main} at
commit~\texttt{229e700}~\cite{L}) using the CORE-MATH \texttt{perf.sh} harness on an AMD Ryzen 9 5900X (Zen~3
architecture, SMT disabled) with random inputs in $[2^{500}, 2^{1000}]$. LLVM libc was compiled with Clang~20.1.2 and
the other libraries, together with the CORE-MATH harness, with GCC~14.2.0, all using \texttt{-O3 -march=native}.
Each wrapper exposes the two-output \texttt{sincos} interface so that both parts of the reduced argument and the integer
quotient $k$ are kept live in throughput and latency modes.
Table~\ref{tab:perf} reports the clock cycles per call given by the harness.
The measured routines do not return exactly the same quantity:
\begin{itemize}
    \item the LLVM libc, CRLIBM, FreeBSD, and glibc measurements return the integer quotient $k$ and a double-double
        reduced argument in radians (for LLVM libc, this includes Step~\ref{alg:step14}, \texttt{quick\_mult} by
        $\pi/128$; for FreeBSD, it includes the wrapper \texttt{\_\_ieee754\_}\allowbreak\texttt{rem\_pio2});
    \item CORE-MATH's \texttt{reduce\_fast} returns $k$ and a double-double reduced argument as a fraction of a full
        turn ($N = 10$), without multiplication by $2\pi$;
    \item CORE-MATH's \texttt{reduce\_large} returns $k$ and a single \texttt{double} as a fraction of a full turn
        ($N = 14$), deferring two double-double multiplications and a Fast2Sum to the bipartite table reconstruction.
\end{itemize}

We do not include Brisebarre et al.~\cite{BDKMR} in Table~\ref{tab:perf}: the ANSI-C implementation referenced in
their paper is no longer available, and their table-based method covers only $|x| < 2^{63}$, deferring larger
arguments, including our benchmark range $[2^{500}, 2^{1000}]$, to Payne--Hanek.
We also omit SLEEF~\cite{SLEEF} from Table~\ref{tab:perf} because its range reduction is implemented primarily as an
internal component inlined within vector routines (relying on a 969-entry, $\approx 30$~KB table), rather than a
standalone scalar function in the benchmark harness. Since glibc~2.44 replaced \texttt{\_\_branred} with
\texttt{\_\_kernel\_rem\_pio2} (Remark~\ref{rem:algo_characteristics}), we expect its large-argument performance to
be close to the FreeBSD column. In LLVM libc~\cite{L}, Algorithm~\ref{alg:main} is implemented in
\texttt{libc/src/\_\_support/math/range\_reduction\_double\_*.h}.

\begin{table}[htbp]
    \caption{Performance Benchmark (Ryzen 9 5900X, Clock Cycles)} \lab{tab:perf}
    \begin{center}
    \setlength{\tabcolsep}{2.5pt}
    \begin{tabular}{ccccccc}
    \toprule
        & & & & \multicolumn{2}{c}{\textbf{CORE-MATH}} & \\
    \cmidrule(lr){5-6}
        Metric & \textbf{FreeBSD} & \textbf{glibc} ($\le 2.43$) & \textbf{CRLIBM} &
        \texttt{fast} & \texttt{large} & \textbf{LLVM libc} \\
    \midrule
        Latency & 438 & 264 & 265 & 68 & 51 & 54 \\
        Recip. Throughput & 379 & 204 & 204 & 36 & 15 & 13 \\
    \bottomrule
    \end{tabular}
    \end{center}
\end{table}

On this machine, Algorithm~\ref{alg:main} achieves the highest throughput (13 cycles reciprocal throughput) among all
tested routines and the lowest latency (54 cycles) among those returning a double-double reduced argument (within
3 cycles of CORE-MATH's single-\texttt{double} \texttt{reduce\_large} at 51 cycles):
\begin{itemize}
    \item Compared to CORE-MATH's \texttt{reduce\_fast} (68 / 36 cycles), which also produces a double-double residual,
        Algorithm~\ref{alg:main} has $1.26\times$ lower latency and $2.77\times$ higher throughput.
    \item Compared to CORE-MATH's newer \texttt{reduce\_large} (51 / 15 cycles), Algorithm~\ref{alg:main} has
        $1.15\times$ higher throughput (13 vs.\ 15 cycles) and comparable latency (54 vs.\ 51 cycles). Relative to
        \texttt{reduce\_fast}, \texttt{reduce\_large} more than doubles throughput ($2.40\times$) and reduces latency
        by $1.33\times$ by replacing 192-bit integer arithmetic and \texttt{set\_dd} normalization with a 128-bit
        multiplication and a single integer-to-\texttt{double} conversion, returning only a single \texttt{double}
        residual $r$ and deferring two double-double multiplications and a Fast2Sum to the bipartite table
        reconstruction. Even though Algorithm~\ref{alg:main} computes a full double-double reduced argument in radians
        (including the fourth limb $D_3$ and Step~\ref{alg:step14}, \texttt{quick\_mult} by $\pi/128$) with error
        $< 2^{-110}$, it exploits instruction-level parallelism across the four floating-point pipes
        (Figure~\ref{fig:pipeline}): the middle and tail limbs execute concurrently with the quotient and high part
        $y_h$, so computing the complete double-double adds only 3 cycles of latency over \texttt{reduce\_large}'s
        single serial dependency chain across floating-point-to-integer and integer-to-floating-point transfers.
    \item Compared to the multi-word routines of FreeBSD, glibc ($\le 2.43$), and CRLIBM (264--438 cycles latency,
        204--379 cycles reciprocal throughput), Algorithm~\ref{alg:main} has $4.9-8.1\times$ lower
        latency and $15.7-29.2\times$ higher throughput.
\end{itemize}
These measurements are for a single microarchitecture (Zen~3) and compiler configuration, and the ratios may differ
on other processors or compilers. For instance, with GCC~14.2.0 the 128-bit-integer-to-\texttt{double} conversion in
\texttt{reduce\_large} is a single \texttt{vcvtsi2sdq} instruction, whereas Clang~20.1.2 emits a call to the
\texttt{\_\_floattidf} runtime routine.

\textbf{SIMD Vectorization and Gather Considerations:}
We have not implemented or measured a SIMD version of Algorithm~\ref{alg:main}, but its structure is amenable to
SIMD vectorization: it is branch-free, and all of its steps, including the integer and bit operations of
Steps~1--3, map to SIMD operations without leaving vector registers. In SIMD code, data-dependent branches require
lane masking or blending of both sides. When the exponents differ across lanes, Step~4 loads the four limbs
$D_0, \ldots, D_3$ of a different table row (32 bytes) for each lane, so a vectorized version needs either one gather
per limb (e.g., \texttt{vgatherdpd} in AVX2/AVX-512) or per-lane row loads followed by a transpose. The cost of
gathers depends on the microarchitecture; for instance, they are relatively slow on Zen~3, and the 2023 microcode
mitigations for Gather Data Sampling made them slower on affected Intel processors. The size of the table also
matters for repeated calls: our table is 2~KB, while SLEEF's is $\approx 30$~KB, compared to typical L1 data caches
of 32--48~KB. Finally, Algorithm~\ref{alg:main} only covers $|x| \ge 2^{16}$, so a vectorized
\texttt{sin}/\texttt{cos} still needs to blend its result with a Cody--Waite reduction when lanes contain arguments
of mixed magnitudes.

\section{Concluding Remarks} \lab{sec:Rem}

We have presented a branch-free variation of the Payne--Hanek range reduction algorithm using floating-point
arithmetic, which is used in LLVM libc~\cite{L}. For binary64, Algorithm~\ref{alg:main} groups exponents in blocks
of $M = 16$, so that a 2~KB table with four limbs per entry covers all inputs with $|x| \ge 2^{16}$. The reduced
argument satisfies $|\hat{u} - u| < 2^{-107} + 2^{-122}$ and $|\hat{y} - y| < 2^{-110}$
(Corollary~\ref{cor:llvm_params}), which meets the worst-case budget of Section~\ref{sec:WorstCase}. On an AMD Ryzen~9
5900X, it takes 54 cycles of latency and 13 cycles of reciprocal throughput (Section~\ref{sec:perf}).

\begin{enumerate}
    \item \textbf{Input Domain Extension}: For $2^{16} \le |x| < 2^{62}$, the table index $i$ is negative, $i \in
        \{-3, -2, -1\}$. The corresponding entries carry a constant prefix $2^{-2}$ in $D_0$, whose contribution
        $32 x_r$ vanishes modulo $2^{N+1} = 256$ since $\uls{x_r} \ge 2^{10}$. This extends the domain down to $|x|
        \ge 2^{16}$ without branching (Section~\ref{sec:LUT}). For $|x| < 2^{16}$, Cody--Waite reduction is used.
    \item \textbf{Rounding Modes and Portability}: Under directed rounding modes ($\operatorname{RU}$,
        $\operatorname{RD}$, $\operatorname{RZ}$), both Fast2Sum operations remain exact by~\cite[Theorem~4]{PLN26},
        and $|\hat{u} - u| < 2^{-105} + 2^{-121}$ when $\hat{k}$ uses static $\operatorname{RN}$ rounding and
        $|\hat{u} - u| < 2^{-104} + 2^{-121}$ with dynamic \texttt{rint} rounding (Section~\ref{sec:rounding}). On
        targets without hardware FMA, the products remain exact in all rounding modes by
        Proposition~\ref{prop:dekker_exact}, using 51-bit limbs and the standard splitting constant $C = 2^{27}+1$
        ($s = 27$) for both operands.
    \item \textbf{Single-Stage and Correctly Rounded Functions}: In single-stage implementations,
        Algorithm~\ref{alg:main} provides the complete range reduction: since it meets the worst-case budget, no
        cancellation detection or accurate path is needed. In two-stage correctly rounded libraries,
        Algorithm~\ref{alg:main} serves as the fast path, with $|\hat{u} - u| < 2^{-107} + 2^{-122}$ on $|u| \le
        1/2 + 2^{-19}$. With three limbs (Algorithm~\ref{alg:three}), $|\hat{u} - u| < 2^{-72} + 2^{-73}$, which
        meets the fast-path budget $2^{-67}$ of Section~\ref{sec:WorstCase} (Corollary~\ref{cor:three_limb}). When
        an input is exceptionally close to a root ($\hat{k} \equiv 0 \pmod{128}$ for $\sin$ or $\hat{k} \equiv 64
        \pmod{128}$ for $\cos$, with $|u| \ll 1$) or to a rounding boundary, and Ziv's rounding test fails, the
        accurate path can reuse the quotient $\hat{k}$ and the exact partial products without recomputing them.
        LLVM libc re-accumulates them together with the fourth-limb product $\circ(x_r D_3)$ in 128-bit arithmetic;
        alternatively, a multi-word reduction can be invoked.
    \item \textbf{Other Formats and SIMD}: The same construction applies to binary32 with $p = p_c = 24$, $N = 3$,
        $M = 8$, and a 14-entry table (0.22~KB) covering $|x| \ge 2^{18}$ (Section~\ref{sec:LUT}). Since
        Algorithm~\ref{alg:main} is branch-free, it is amenable to SIMD vectorization, but we have not implemented
        or measured a vectorized version (Section~\ref{sec:perf}).
    \item \textbf{Future Work}: We plan two follow-up studies. First, the idea of exponent grouping and limb
        alignment also applies to integer-only implementations (Remark~\ref{rem:integer_table}). Such a routine is
        already used in LLVM libc~\cite{L} for double-precision $\sin$/$\cos$ on memory-con\-strained and
        soft-float targets: it uses a 167-byte table of $2/\pi$ and 128-bit fixed-point arithmetic, and does not
        require floating-point hardware. We will present its formulation and analysis in a follow-up paper.
        Second, we will study the trade-offs between the number of limbs, the exponent grouping $M$ (for example,
        $M = 32$ with four limbs under the fast-path budget, Remark~\ref{rem:fastpath}), and the fraction of inputs
        that reach the accurate path.
\end{enumerate}

\appendix

\section{Three-Limb Variant} \lab{sec:three_limb}

For the fast path of a correctly rounded implementation, the budget $\varepsilon = 2^{-67}$ of
Section~\ref{sec:WorstCase} is met by the first three limbs of the table in Section~\ref{sec:LUT}
(Remark~\ref{rem:fastpath}). With three limbs, the omitted tail $x_r (D_3 + D_4 + \cdots)$ is up to $0.91 \cdot
2^{-72}$, while the low part of the exact product $x_r D_2$ is at most $2^{-74}$. So $x_r D_2$ does not need to be
computed exactly, and it can be merged with $b_\ell$ in a single fused multiply-add. Algorithm~\ref{alg:three} is
the resulting three-limb variant of Algorithm~\ref{alg:main}, and Figure~\ref{fig:pipeline3} shows its dataflow.
It differs from Algorithm~\ref{alg:main} in the following:
\begin{itemize}
    \item each table entry stores $D_0, D_1, D_2$, so the table takes $64 \cdot 3 \cdot 8$ bytes $= 1.5$~KB;
    \item steps~7 and~10--13 of Algorithm~\ref{alg:main} (the exact product $x_r D_2$, two Fast2Sum additions, the
        tail FMA, and two additions) are replaced by one FMA $q = \circ(b_\ell + x_r D_2)$ and one Fast2Sum.
\end{itemize}
Algorithm~\ref{alg:three} is not part of LLVM libc at commit~\texttt{229e700}~\cite{L}, and we have not benchmarked
it. The analysis below assumes round-to-nearest ($\operatorname{RN}$) mode.

\begin{algorithm}[htbp]
    \caption{Three-Limb Branch-Free Floating-Point Range Reduction} \lab{alg:three}
    \begin{algorithmic}[1]
        \REQUIRE $x \in \F$ with $|x| \ge 2^{e_{x,\min}}$ (e.g., $|x| \ge 2^{16}$ for double precision)
        \ENSURE $\hat{k} \bmod 2^{N + 1}$, and double-double $\hat{u} = u_h + u_\ell$ (or $\hat{y} = y_h + y_\ell$)
        \STATE Extract unbiased exponent $e_x \ge e_{x,\min}$ from $x = (-1)^{s_x} 2^{e_x} m_x$
        \STATE $i \leftarrow \floor{\frac{e_x - p - g}{M}}$ \COMMENT{Table index $i \in [i_{\min}, i_{\max}]$}
        \STATE Scale input $x_r \leftarrow x \cdot 2^{-M i}$ via exponent adjustment
        \STATE $D_0, D_1, D_2 \leftarrow \operatorname{PI\_INV}[i]$ \COMMENT{Three limbs per entry}
        \STATE $a_h + a_\ell \leftarrow x_r \cdot D_0$ \COMMENT{Exact product, $\ulp{a_h} \ge 2^{N+1}$}
        \STATE $b_h + b_\ell \leftarrow x_r \cdot D_1$ \COMMENT{Exact product, $\ulp{b_h} \le 1/4$}
        \STATE $\hat{k} \leftarrow \nearestint{a_\ell \oplus b_h}$ \COMMENT{Integer quotient}
        \STATE $v \leftarrow (a_\ell \ominus \hat{k}) \oplus b_h$ \COMMENT{Exact operations by
            Theorem~\ref{thm:FirstSteps}}
        \STATE $q \leftarrow \circ(b_\ell + x_r \cdot D_2)$ \COMMENT{Fused multiply-add}
        \STATE $u_h + u_\ell \leftarrow \fts{v, q}$ \COMMENT{Exact sum since $\uls{v} \ge \ulp{q}$~\cite{JZ25}}
        \STATE $y_h + y_\ell \leftarrow (u_h + u_\ell) \cdot \frac{\pi}{2^N}$ \lab{alg3:stepy}
            \COMMENT{Optional scaled output}
        \RETURN $\hat{k} \bmod 2^{N+1}$, $\hat{u} = u_h + u_\ell$ (or $\hat{y} = y_h + y_\ell$)
    \end{algorithmic}
\end{algorithm}

\begin{figure}[htbp]
\centering
\begin{tikzpicture}[
    scale=0.85, every node/.style={transform shape},
    box/.style={rectangle, draw, rounded corners=2pt, minimum height=1.7em, minimum width=4.5em,
        text centered, font=\footnotesize},
    tablebox/.style={rectangle, draw, fill=blue!8, rounded corners=2pt, minimum height=1.7em,
        text centered, font=\footnotesize},
    op/.style={circle, draw, fill=orange!15, inner sep=1.5pt, font=\scriptsize},
    res/.style={rectangle, draw, fill=green!10, rounded corners=2pt, minimum height=1.7em,
        text centered, font=\footnotesize\bfseries},
    arr/.style={-stealth, thick, draw=gray!80},
    bus/.style={thick, draw=gray!80}
]
    \node[box, fill=gray!10] (x) at (-1.8, 6.7) {$x$};
    \node[tablebox] (idx) at (1.8, 6.7) {Index $i = \floor{\frac{e_x - 62}{16}}$};

    \node[box, fill=gray!10] (xr) at (-1.8, 5.5) {$x_r = x \cdot 2^{-16i}$};
    \node[tablebox] (lut) at (1.8, 5.5) {$\operatorname{PI\_INV}[i] \to (D_0, D_1, D_2)$};

    \draw[arr] (x) -- (idx);
    \draw[arr] (x) -- (xr);
    \draw[arr] (idx) -- (lut);
    \draw[arr] (idx.south) to[out=210, in=30] (xr.north east);

    \draw[bus] (-4.5, 4.6) -- (2.0, 4.6);
    \draw[bus] (xr.south) -- (-1.8, 4.6);
    \draw[bus] (lut.south) -- (1.8, 4.6);

    \node[box, fill=purple!8] (p0) at (-4.5, 3.7) {$a_h + a_\ell = x_r D_0$};
    \node[box, fill=purple!8] (p1) at (-1.5, 3.7) {$b_h + b_\ell = x_r D_1$};

    \draw[arr] (-4.5, 4.6) -- (p0.north);
    \draw[arr] (-1.5, 4.6) -- (p1.north);

    \node[op] (kadd) at (-3.0, 2.6) {$\oplus$};
    \node[res] (k) at (-4.5, 1.5) {$\hat{k} = \nearestint{a_\ell \oplus b_h}$};
    \node[op] (vsub) at (-2.7, 1.5) {$\ominus$};
    \node[op] (vadd) at (-1.5, 1.5) {$\oplus$};
    \node[box, fill=yellow!15] (v) at (-1.5, 0.4) {$v = (a_\ell \ominus \hat{k}) \oplus b_h$};

    \draw[arr] (p0.south) to[out=270, in=120] (kadd);
    \draw[arr] (p1.south) to[out=270, in=60] (kadd);
    \draw[arr] (kadd) -- (k);
    \draw[arr] (k) -- (vsub);
    \draw[arr] (p0.south) to[out=270, in=150] (vsub);
    \draw[arr] (vsub) -- (vadd);
    \draw[arr] (p1.south) to[out=270, in=90] (vadd);
    \draw[arr] (vadd) -- (v);

    \node[op] (qfma) at (2.0, 2.6) {$\fma$};
    \node[box, fill=yellow!15] (q) at (2.0, 1.4) {$q = \circ(b_\ell + x_r D_2)$};
    \draw[arr] (2.0, 4.6) -- (qfma);
    \draw[arr] (p1.south) to[out=270, in=150] node[pos=0.6, above, font=\footnotesize] {$b_\ell$} (qfma);
    \draw[arr] (qfma) -- (q);

    \node[box, fill=cyan!10] (f2s2) at (0.0, -0.5) {$\fts{v, q} \to (u_h, u_\ell)$};
    \draw[arr] (v.south) to[out=270, in=150] (f2s2.north);
    \draw[arr] (q.south) to[out=270, in=30] (f2s2.north);

    \node[res] (res) at (0.0, -1.7) {$\hat{u} = u_h + u_\ell$};
    \draw[arr] (f2s2) -- (res);

\end{tikzpicture}
\caption{Dataflow pipeline of the three-limb variant (Algorithm~\ref{alg:three}).} \lab{fig:pipeline3}
\end{figure}
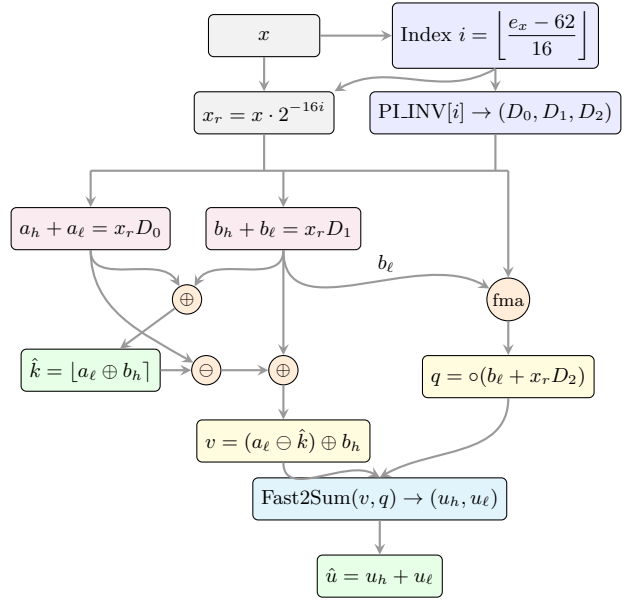

\begin{Cor}[Three-Limb Variant] \lab{cor:three_limb}
    For the parameters of Corollary~\ref{cor:llvm_params} ($p = 53, p_c = 51, N = 7, M = 16, g = 9$, with $i \in
    [-3, 60]$ covering $|x| \ge 2^{16}$), under round-to-nearest mode, Algorithm~\ref{alg:three} satisfies:
    \begin{enumerate}
        \item \textbf{Quotient Separation and Residual Range:}
        Item~(1) of Corollary~\ref{cor:llvm_params} holds: $\Delta_k < 2^{-19}$, $\hat{k}$ is exact whenever
        $\left\{x \cdot \frac{2^N}{\pi}\right\}$ is separated from $\pm 1/2$ by more than $2^{-19}$, and:
        \eq{}{
            \abs{v} \le \frac{1}{2} + 2^{-21}, \quad \abs{u} < \frac{1}{2} + 2^{-19}.
        }

        \item \textbf{Fast2Sum Exactness:}
        The addition $\fts{v, q}$ is exact: $u_h + u_\ell = v + q$.

        \item \textbf{Double-Double Residual Accuracy:}
        The computed residual $\hat{u} = u_h + u_\ell$ satisfies:
        \eq{uBound3}{
            \abs{\hat{u} - u} \le \frac{1}{2}\ulp{q} + \abs{x_r \sum_{j=3}^\infty D_j} < 2^{-72} + 2^{-73}.
        }
        In particular, if $|u| \ge 2^{-15}$ (the fast-path condition of Section~\ref{sec:WorstCase}), then:
        \eq{}{
            \frac{\abs{\hat{u} - u}}{\abs{u}} < 2^{-57} + 2^{-58}.
        }

        \item \textbf{Accuracy of the Scaled Reduced Argument:}
        With the same constant $P_h + P_\ell \approx \pi/128$ as in Corollary~\ref{cor:llvm_params}, the scaled
        reduced argument $\hat{y} = y_h + y_\ell$ in step~\ref{alg3:stepy} satisfies:
        \eq{}{
            \abs{\prt{y_h + y_\ell} - \prt{x - \hat{k} \cdot \frac{\pi}{2^N}} \bmod 2\pi}
                < 2^{-77} + 2^{-78} + 2^{-110}.
        }
    \end{enumerate}
\end{Cor}

\begin{proof}
    Evaluating all 64 table entries with $|x_r| < 2^{78}$ gives:
    \eq{}{
        \abs{x_r D_2} < 0.9 \cdot 2^{-20}, \quad \abs{x_r \sum_{j=3}^\infty D_j} < 0.91 \cdot 2^{-72}.
    }

    \emph{Proof of (1):}
    Steps~1--8 of Algorithm~\ref{alg:three} coincide with steps~1--6, 8, and~9 of Algorithm~\ref{alg:main}, and
    they do not involve $D_2$. The quantity $\Delta_k$ in~\eqr{DeltaKDef} only depends on these steps and on the
    exact limbs of the expansion, not on which limbs are stored. So items~(1) and~(2) of
    Theorem~\ref{thm:correctness} apply, and the proof of Corollary~\ref{cor:llvm_params}~(1) gives the same bounds
    $\Delta_k < 2^{-19}$, $|v| \le 1/2 + 2^{-21}$, and $|u| < 1/2 + 2^{-19}$.

    \emph{Proof of (2):}
    From Theorem~\ref{thm:FirstSteps}, $v$ is an integer multiple of $\ulp{b_h}$, so $\uls{v} \ge \ulp{b_h}$ if $v
    \ne 0$. Since $D_1$ is rounded to $p_c = 51$ bits, $|D_2| \le 2^{-51} \ufp{D_1}$. Together with $|x_r| < 2
    \ufp{x_r}$ and $\ufp{x_r} \ufp{D_1} \le \ufp{b_h}$, this gives:
    \eq{}{
        \abs{x_r D_2} < 2^{-50} \ufp{x_r} \ufp{D_1} \le 2^{-50} \ufp{b_h} = 4 \ulp{b_h}.
    }
    Since $|b_\ell| \le \frac{1}{2}\ulp{b_h}$, we have:
    \eq{}{
        \abs{b_\ell + x_r D_2} < \frac{1}{2}\ulp{b_h} + 4\ulp{b_h} < 2^3 \ulp{b_h},
    }
    where rounding preserves the bound because $\frac{1}{2}\ulp{b_h} + 4\ulp{b_h} \in \F$. Therefore:
    \eq{}{
        \ulp{q} \le 2^{-52} \cdot 2^2 \ulp{b_h} = 2^{-50} \ulp{b_h} < \ulp{b_h} \le \uls{v},
    }
    and $\fts{v, q}$ is exact by~\cite[Theorem~1]{JZ25}. If $v = 0$, exactness holds trivially.

    \emph{Proof of (3):}
    Since $a_h \equiv 0 \pmod{2^{N+1}}$ and $v = a_\ell - \hat{k} + b_h$ exactly, the exact residual satisfies:
    \eq{}{
        u = v + b_\ell + x_r D_2 + x_r \sum_{j=3}^\infty D_j.
    }
    By item~(2), $\hat{u} = u_h + u_\ell = v + q$, so:
    \eq{}{
        \hat{u} - u = \prt{q - (b_\ell + x_r D_2)} - x_r \sum_{j=3}^\infty D_j,
    }
    which gives the first inequality of~\eqr{uBound3}. From the proof of Corollary~\ref{cor:llvm_params}~(1),
    $|b_\ell| \le 2^{-22}$, so:
    \eq{}{
        \abs{b_\ell + x_r D_2} < 2^{-22} + 0.9 \cdot 2^{-20} < 2^{-19}.
    }
    Rounding preserves this bound, so $\ufp{q} \le 2^{-20}$ and:
    \eq{}{
        \frac{1}{2}\ulp{q} \le 2^{-53} \cdot 2^{-20} = 2^{-73}.
    }
    Adding the omitted tail gives:
    \eq{}{
        \abs{\hat{u} - u} < 2^{-73} + 0.91 \cdot 2^{-72} < 2^{-72} + 2^{-73}.
    }
    Finally, if $|u| \ge 2^{-15}$, then:
    \eq{}{
        \frac{\abs{\hat{u} - u}}{\abs{u}} < \frac{2^{-72} + 2^{-73}}{2^{-15}} = 2^{-57} + 2^{-58}.
    }

    \emph{Proof of (4):}
    We bound the three terms of~\eqr{yBoundGen} as in the proof of Corollary~\ref{cor:llvm_params}~(4). Since
    $\pi/128 < 2^{-5}$, the propagated error satisfies:
    \eq{}{
        \abs{\hat{u} - u} \cdot \frac{\pi}{128} < \prt{2^{-72} + 2^{-73}} \cdot 2^{-5} = 2^{-77} + 2^{-78}.
    }
    Since $|v + q| < 1$, we have $|u_h| < 1$, and the Fast2Sum in step~10 gives $|u_\ell| \le \frac{1}{2}\ulp{u_h}
    \le 2^{-54}$. So the double-double multiplication error is still $\epsilon_{\text{mult}} < 2^{-111}$, and the
    constant error is still below $2^{-115}$. Summing these three contributions yields:
    \eq{}{
    \begin{aligned}
        \abs{\hat{y} - y} &< 2^{-77} + 2^{-78} + 2^{-111} + 2^{-115} \\
            &< 2^{-77} + 2^{-78} + 2^{-110}.
    \end{aligned}
    }
\end{proof}

\end{document}